\documentclass[journal]{IEEEtran}
\usepackage[T1]{fontenc}
\usepackage[utf8]{inputenc}
\usepackage{amsmath,amssymb,amsthm,bm}
\usepackage{graphicx}
\usepackage[caption=false]{subfig}
\usepackage{booktabs}
\usepackage{cite}
\usepackage{tikz}
\usetikzlibrary{calc,decorations.pathreplacing}
\usepackage{algorithm}
\usepackage{algpseudocode}
\usepackage{stfloats}            
\graphicspath{{Figures/}}

\newtheorem{theorem}{Theorem}
\newtheorem{proposition}{Proposition}
\newtheorem{lemma}{Lemma}
\newtheorem{corollary}{Corollary}
\newtheorem{assumption}{Assumption}
\theoremstyle{remark}
\newtheorem{remark}{Remark}

\DeclareMathOperator{\Si}{Si}

\newcommand{\bigO}{\mathcal{O}}
\newcommand{\lilo}{\mathrm{o}}
\newcommand{\softO}{\widetilde{\mathcal{O}}}
\newcommand{\NF}{N_{\mathrm{F}}}
\newcommand{\NFs}{N_{\mathrm{F}}^{\star}}
\newcommand{\NFe}{N_{\mathrm{F}}^{\mathrm{eff}}}
\newcommand{\gfoc}{G_{\mathrm{foc}}}
\newcommand{\gff}{G_{\mathrm{FF}}}
\newcommand{\geig}{G_{\mathrm{eig}}}
\newcommand{\gcf}{G_{\mathrm{cf}}}
\newcommand{\wt}{\bm{w}_{\mathrm{t}}}
\newcommand{\wrx}{\bm{w}_{\mathrm{r}}}
\newcommand{\bH}{\bm{H}}
\newcommand{\suppl}{the supplementary material}

\begin{document}

\title{From Focusing to Con-Focusing:
Optimal Power Transfer in Line-of-Sight Near-Field MIMO}

\author{Marouan~Mizmizi,~\IEEEmembership{Member,~IEEE,}
and~Sanzhar~Yergaliev,~\IEEEmembership{Student~Member,~IEEE}%
\thanks{Manuscript draft, June 2026.}%
\thanks{The authors are with the Department of Electronics, Information
and Bioengineering (DEIB), Politecnico di Milano, 20133 Milan, Italy
(e-mail: marouan.mizmizi@polimi.it; sanzhar.yergaliev@polimi.it).}}

\maketitle

\begin{abstract}
Beamfocusing is the established near-field strategy for a large array
serving a single-antenna user. We consider the single-user
line-of-sight MIMO link, free of multipath, in which the user, too,
carries an extended aperture, and show that the focusing prescription
inverts: beyond a modest Fresnel number, focusing on the user is
outperformed by far-field steering. Under fully analog, unit-modulus
beamforming, we derive closed-form power gains for \emph{focusing}
(each aperture phase-matched to the other's center) and for
\emph{steering} (a planar phase ramp) in the Fresnel regime, and prove
that their comparison is governed by two dimensionless quantities: the
link Fresnel number, the product of the two aperture lengths normalized
by wavelength and link distance, and the aperture ratio, irrespective
of how many elements discretize the apertures. For equal apertures the
two gains cross exactly once, at the universal value $1.947$; beyond
it, focusing loses ten dB per decade of Fresnel number, and the
advantage celebrated in the MISO literature survives only as the
receive aperture vanishes. We then derive the strategy that is
order-optimal at every Fresnel number, \emph{con-focusing}: both
apertures aim at the common point from which they subtend equal angles.
It attains the rank-one eigenbound in leading constant, needs no
channel knowledge, degenerates to plain steering for equal apertures,
and is acquirable within one beam-refinement round with no geometry
exchange between the terminals.
\end{abstract}


\section{Introduction}\label{sec:intro}

\IEEEPARstart{C}{onsider} a line-of-sight link between two antenna
arrays, free of multipath, each carrying an arbitrary number of
elements. When the apertures are small or widely separated their
wavefronts are essentially planar, and classical beam steering is
optimal. As the apertures grow and the link shortens, it enters the radiative
near field, where the planar-wavefront model behind classical steering
no longer holds. This is the regime of the mmWave and sub-THz carriers
envisioned for next generation of wireless networks~\cite{Rappaport2019,Saad2020,Tataria2021,Akyildiz2022}:
a few-centimeter aperture already packs tens to hundreds of
elements~\cite{Marzetta2010,Larsson2014}, and the Rayleigh distance
recedes to tens of meters~\cite{Selvan2017,BjornsonSanguinetti2020}.

\begin{figure}[!b]
\centering
\begin{minipage}{0.49\columnwidth}\centering
\includegraphics[width=\linewidth]{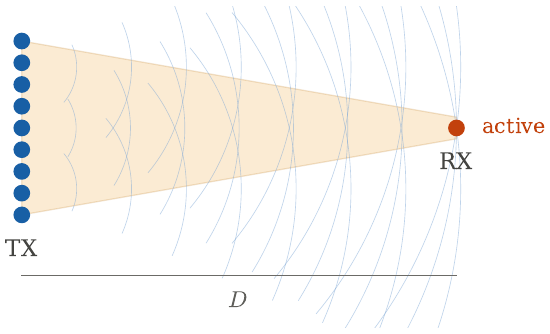}\\[-2pt]
{\footnotesize (a) MISO}
\end{minipage}\hfill
\begin{minipage}{0.49\columnwidth}\centering
\includegraphics[width=\linewidth]{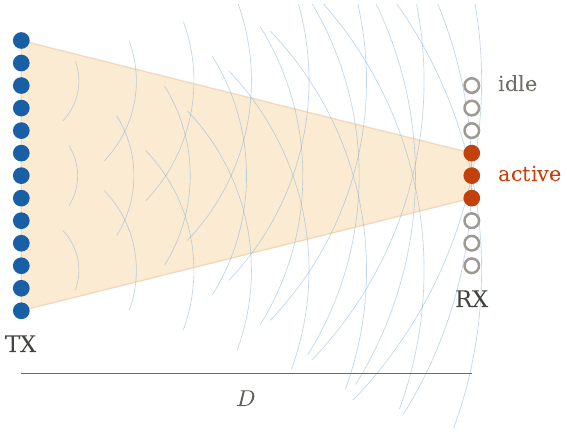}\\[-2pt]
{\footnotesize (b) MIMO}
\end{minipage}
\caption{Why focusing dies in MIMO. (a)~For a single-antenna user the
focal spot serves the one active element: the MISO regime, where
focusing is never inferior. (b)~When the user, too, carries a large
aperture, the same spot reaches only the few central elements while
the rest sit idle once $\NF > 1$.}
\label{fig:concept}
\end{figure}

The MIMO literature has drawn from this premise a now-standard
conclusion: within the Rayleigh distance, steering should be
replaced by \emph{focusing}, i.e., phase profiles matched to the
spherical wavefront converging on the intended
user~\cite{CuiDai2022,ZhangEldar2022,LuZeng2024,CuiDaiMag2023,%
ZhangMag2023,LiuSurvey2024,LuZengTWC2022}. The supporting
analysis, however, is almost invariably carried out for a large array
serving single-antenna terminals: a multiple-input
\emph{single}-output (MISO) geometry, even when the system is labeled
MIMO because many such terminals are served at
once~\cite{ZhangEldar2022,WuDai2023}. The comprehensive tutorial
reviews are explicit on this point: their demonstrations of a
near-field focusing gain, and of steering
underperforming as the array grows, are carried out for a single
large array serving single-antenna users~\cite{LiuSurvey2024,%
LuZeng2024}, never for a link between two extended apertures. For a
single-antenna user the conclusion is correct: focusing is never
inferior to steering at any distance. This is the
small-receive-aperture corner of a more general link. We generalize
near-field beamfocusing to the single-user MIMO case in which the
user, too, carries an extended aperture, the two ends related by their
aperture ratio $\rho$. The question is
how much focusing on the user is then worth.

The answer is: very little, and we can say exactly how little.
Consider a transmit aperture of length $L_{\mathrm{t}}$ and a receive
aperture of length $L_{\mathrm{r}}$, facing each other at distance
$D$, at carrier wavelength $\lambda$. The phase profile that focuses
the transmitter onto the receiver's center creates, at range $D$, a
focal spot of width $\lambda D/L_{\mathrm{t}}$, the
diffraction-limited image of the transmit aperture. As long as the
receive aperture sits inside the spot, every receive element is
served coherently and the MISO intuition applies. Once
$L_{\mathrm{r}}$ outgrows the spot, the focused field simply misses
most of the receiver, and adding receive elements buys nothing
(Fig.~\ref{fig:concept}). The condition
``$L_{\mathrm{r}}$ exceeds the spot width'' is set by
$L_{\mathrm{t}}L_{\mathrm{r}}/(\lambda D) > 1$. The two dimensionless
quantities that govern the entire problem are the \emph{link Fresnel
number} $\NF$ and the \emph{aperture ratio} $\rho$,
\begin{equation}\label{eq:NF-def}
\NF \triangleq \frac{\rho\, d_{\mathrm{F}}}{2D}
= \frac{L_{\mathrm{t}}L_{\mathrm{r}}}{\lambda D},
\qquad
\rho \triangleq \frac{L_{\mathrm{r}}}{L_{\mathrm{t}}} \le 1 ,
\end{equation}
where $d_{\mathrm{F}} \triangleq 2L_{\mathrm{t}}^{2}/\lambda$ is the
Fraunhofer distance of the larger aperture. Written this way, $\NF$ has
an immediate geometric reading: it is the aperture ratio $\rho$ times
the depth of the link inside the Fraunhofer distance, $d_{\mathrm{F}}/D$,
so the near field $\NF \gtrsim 1$ is exactly the regime
$D \lesssim \rho\, d_{\mathrm{F}}/2$, in which the link lies within (the
fraction $\rho/2$ of) the Fraunhofer distance of the larger aperture. 
Steering degrades far more gracefully than
focusing as $\NF$ grows, and a crossover follows. As we shall see,
the whole trade-off collapses onto the pair $(\NF, \rho)$:
the crossover sits at the universal constant $\NFs = 1.947$ for
equal apertures, and at an aperture-dependent $\NFs(\rho)$ in general,
tending to infinity as $\rho \to 0$; in the deep near field steering
wins by the factor $\rho\NF$ in SNR, and
the MISO regime is recovered only in the limit $\rho \to 0$. This limit
is one of vanishing receive aperture, not of a single antenna: what
governs the trade-off is aperture extent, not the number of elements
that discretize it. A small receive aperture stays in the MISO regime
whether it carries one element or many, and an extended one leaves it
even with a single port.

The problem of
maximizing the power coupled between two facing apertures, neither in
the far field of the other, is well investigated. Kay~\cite{Kay1960} derived
the near-field power-transfer bound, the continuous-aperture ancestor
of the eigen-beamforming bound used here, and remarked, without
development, that spherical phase conjugation ``may no longer be
optimum in the nearer part of the Fresnel zone.'' Borgiotti~\cite{Borgiotti1966}
proved that the jointly optimal illuminations are the confocal prolate
modes, establishing optimality over the confocal family as a whole;
that its \emph{equal-subtense} member is the unique optimum, and that a
\emph{unit-modulus} pair attains it in the MIMO link, is established
here. His numerical tables already bracket an inversion between focused
and unfocused square apertures, consistent with the closed forms of
Section~\ref{sec:gains}. What this early work lacks is the
parametrization of the trade-off, the location, uniqueness, and
universality of the crossover, its asymptotic laws, the role of
aperture asymmetry, and any implication for MIMO arrays and beam
training, concepts that did not exist at the time.

The eigenproblem of Kay and Borgiotti was later rediscovered in the
language of communication modes~\cite{Miller2000} and brought into
the wireless literature in the context of holographic and
reconfigurable surfaces~\cite{Dardari2020,Huang2020,DiRenzo2020}.
Decarli and Dardari~\cite{DecarliDardari2021} showed that
\emph{phase-only} focusing generates quasi-orthogonal communication
modes whose number matches the optimal one when a large intelligent
surface serves a much smaller antenna, a conclusion that the
$(\NF,\rho)$ map derived here explains and delimits: it is the
$\rho \ll 1$ region of the map. A parallel LoS MIMO thread designs
array geometries that maximize multiplexing, such as the
antenna-separation-product criterion of B{\o}hagen
\emph{et al.}~\cite{Bohagen2007} (see
also~\cite{Gesbert2002,Bohagen2009}), the beamspace CAP-MIMO
architecture~\cite{SayeedBehdad2011}, and hybrid wide-aperture
designs~\cite{YunRhoChoi2024}; it does not, however, compare
single-stream focusing against steering under unit-modulus
constraints, which is the comparison resolved here. Finally, the
XL-MIMO literature has developed dictionaries of point-focusing
atoms (the polar-domain codebook of Cui and Dai~\cite{CuiDai2022}
and its refinements~\cite{WeiDai2022,CuiRainbow2023,LiuSurvey2024}),
under the premise that focusing is required throughout the
Rayleigh region. The early inversions appear to have gone unnoticed
there, and we are not aware of any XL-MIMO work observing that, for
aperture-symmetric LoS links beyond a modest Fresnel number, steering outperforms focusing on the user.

Closest to our setting is~\cite{Shi2025}, which addresses the same
double-sided near-field geometry. Their goal is different: spatial
multiplexing over the link's near-field degrees of freedom, which they
pursue by selecting beamfocusing codewords from a dictionary.
Restricted to a single beam, their construction amounts to the same
focal correction as our con-focusing law. They obtain it, however,
indirectly: as a correction estimated from full channel state
information, itself acquired through a costly compressive-sensing
procedure. We instead state the confocality condition analytically and
attain it directly, from power measurements alone.

This paper provides the theory of this trade-off: the exact location,
uniqueness, and universality of the crossover, its asymptotic laws,
the role of aperture asymmetry, and the consequences for arrays
and beam training. The early numerical values fall
out of our closed forms as special cases.
Establishing these results is the first
purpose of this paper.

The second purpose is constructive. If focusing on the user is only
the $\rho \to 0$ prescription and plain steering only the $\rho = 1$
one, what is the unit-modulus optimum in between? We show that it is a
\emph{con-focusing law}:
each aperture applies the quadratic phase focused not on the opposite
array but on the common axial point from which both apertures subtend
equal angles: the transmitter at range $D/(1+\rho)$, the receiver at
range $\rho D/(1+\rho)$. This is the point at which the converging transmit
beam exactly fills the receive aperture. The title is therefore to be
read precisely: what fails is focusing \emph{on} the user, while the
optimum stays confocal \emph{with} it. A twin solution places the
common point beyond the receiver and, for equal apertures,
degenerates into plain steering, whereas for $\rho \to 0$ the law
collapses onto the user and recovers the MISO prescription. This
geometry-only, channel-independent pair is asymptotically optimal
among all rank-one schemes: the unit-modulus phase loses nothing to the
prolate amplitude tapers (Borgiotti~\cite{Borgiotti1966}) that phase
shifters could not realize anyway. Steering is its
$r \to \infty$ limit, already within $10\log_{10}(1/\rho)$~dB; the law
tolerates ranging errors up to the classical depth of focus and is
acquirable blindly when no ranging is available.

The contributions are as follows.
\begin{enumerate}
\item We derive closed-form unit-modulus gains for focusing and steering on the Fresnel kernel, and we prove
that for equal apertures their ratio crosses unity exactly once, at
$\NFs = 1.947$, independent of antenna counts, wavelength, and
distance separately. In the deep near field the focusing power gain
decays as $1/\NF^{2}$ while the steering power gain decays only as
$1/\NF$: the SNR gap widens by $10$~dB per decade.
\item We construct the complete operating-region boundary in the
$(\NF,\rho)$ plane, which reconciles the classical MISO result
($\rho \to 0$), the LIS-to-small-device results
of~\cite{DecarliDardari2021}, and the symmetric-aperture inversion
of~\cite{Borgiotti1966} in one picture, and we prove that the
focusing advantage is an \emph{asymmetry} phenomenon.
\item We propose the con-focusing law and show that this
geometry-only pair is asymptotically optimal among rank-one schemes,
attaining the rank-one eigenbound in leading order as the Fresnel
number grows; amplitude tapering and channel knowledge then add only a
vanishing margin. The law has two degenerate solutions, the two centers
of similitude of the aperture pair, one of which reduces to plain
steering for equal apertures; its tolerance to ranging errors is set by
the classical depth of focus of the joint aperture. 
\item We reduce blind acquisition to the estimation of a single
scalar, the range $D$, and develop an adaptive maximum-likelihood
protocol that estimates $D$ from received power alone and reaches the
optimum in a small, adaptively placed budget of probes, with no
ranging and no geometry exchange.
\item We prove that the entire ($N_F, \rho$) map is covariant under array
rotations: tilting either array merely replaces the apertures by
their projections onto the broadside plane, including a roll effect
for non-coplanar linear arrays; we quantify the discrete-array sampling
conditions; and we validate every analytical claim against the exact
spherical channel, with reproducible code.
\end{enumerate}

The rest of the paper is organized as follows.
Section~\ref{sec:model} states the model.
Section~\ref{sec:gains} derives the closed-form gains and the
crossover. Section~\ref{sec:confocus} establishes the optimality of
the con-focusing law and its depth of focus, and extends it from
boresight to tilted and rolled arrays.
Section~\ref{sec:acquisition} develops the acquisition protocol.
Section~\ref{sec:numerical}
presents the numerical campaign, and Section~\ref{sec:conclusion}
concludes. All proofs are collected in \suppl.

\emph{Notation.} Boldface lowercase and uppercase letters denote
vectors and matrices (e.g.\ $\bm{w}$, $\bm{H}$); $(\cdot)^{H}$ is the
Hermitian transpose, $|\cdot|$ the modulus, $\|\cdot\|_{1}$ the
$\ell_{1}$ norm, $\sigma_{\max}(\cdot)$ the largest singular value,
and $\triangleq$ a definition. To compare two positive quantities as
the Fresnel number grows we use the Landau symbols: $f = \bigO(g)$
means $f \le c\,g$ for some constant $c$ (upper bound),
$f = \lilo(g)$ means $f/g \to 0$, $f = \Theta(g)$ means
$c_{1}g \le f \le c_{2}g$ (a tight two-sided order), and
$\softO(\cdot)$ is $\bigO(\cdot)$ up to logarithmic factors;
$a \asymp b$ and $a \lesssim b$ denote equality and inequality up to
an absolute constant. 

\section{System Model and Problem Statement}\label{sec:model}

\begin{figure}[!b]
\centering
\includegraphics[width=0.82\columnwidth]{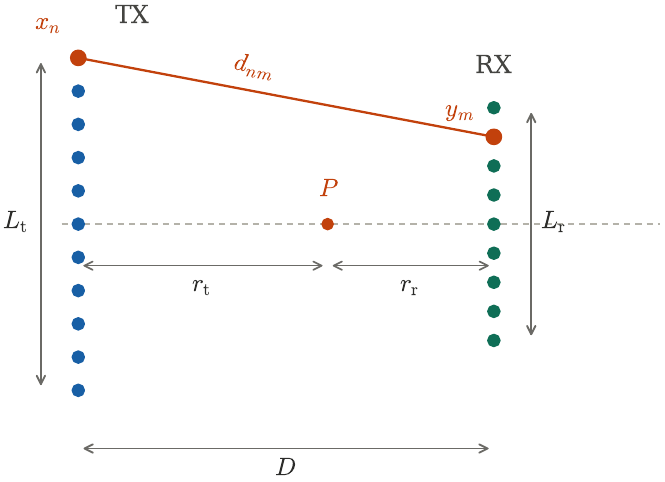}
\caption{Link geometry. Two parallel ULAs of lengths $L_{\mathrm{t}}$
and $L_{\mathrm{r}}$ at separation $D$. A generic transmit element at
$x_n$ and receive element at $y_m$ are joined by the path $d_{nm}$ of
\eqref{eq:dist}. A focused pair \eqref{eq:weights} aims at axial depths
$r_{\mathrm{t}}$ (from TX) and $r_{\mathrm{r}}$ (from RX); when
$r_{\mathrm{t}} + r_{\mathrm{r}} = D$ the two foci meet at a common
point $P$ on the axis. The dimensionless numbers
\eqref{eq:NF-def} and \eqref{eq:NF} are built from
$L_{\mathrm{t}}, L_{\mathrm{r}}, D, \lambda$; the residual phase left by
the pair is measured by the curvatures $A, B$ of \eqref{eq:AB}.}
\label{fig:geometry}
\end{figure}

Consider a LoS link between two parallel uniform linear arrays (ULAs)
in boresight configuration, as in Fig.~\ref{fig:geometry}. The transmit
(TX) array has $N$ elements at positions $x_n \in [-L_{\mathrm{t}}/2,
L_{\mathrm{t}}/2]$; the receive (RX) array has $M$ elements at
$y_m \in [-L_{\mathrm{r}}/2, L_{\mathrm{r}}/2]$ on a parallel line at
distance $D$. The wavelength is $\lambda$ and $k = 2\pi/\lambda$. A
transmit element at $x_n$ and a receive element at $y_m$ are separated
by
\begin{equation}\label{eq:dist}
d_{nm} = \sqrt{D^{2} + (x_n - y_m)^{2}} .
\end{equation}
With a pure-LoS channel the propagation between the two elements is a
single spherical wave, so the channel matrix has unit-modulus entries
\begin{equation}\label{eq:channel}
H_{mn} = e^{-jk d_{nm}} ,
\end{equation}
the amplitude variation $1/d_{nm}$ being deferred to
Section~\ref{sec:gains}, where it is shown uniformly negligible.

The pure-LoS assumption is standard at high frequencies. Measurement campaigns at and above $100$~GHz consistently
report sparse multipath: reflection and scattering losses leave the few
specular components tens of decibels below the direct path, and
blockage removes them
altogether~\cite{PriebeKurner2013,HanAkyildiz2015,JuRappaport2021}. For
the short-range, large-aperture links considered here the direct path
is therefore not merely dominant but, for design purposes, the channel.
The same model is the reference setting of LoS MIMO array
design~\cite{Bohagen2007,SayeedBehdad2011} and of the near-field
XL-MIMO analyses~\cite{CuiDai2022,ZhangEldar2022,LuZengTWC2022}.

\subsection{Beamforming and array gain}

Both terminals employ single-RF-chain analog beamforming with
per-element unit-modulus weights, $|w_{\mathrm{t},n}| = 1/\sqrt{N}$ and
$|w_{\mathrm{r},m}| = 1/\sqrt{M}$, the standard constraint of
phase-shifter front-ends at these
frequencies~\cite{Heath2016,ElAyach2014}. Collecting the weights in
$\wt$ and $\wrx$, with transmit power $P$ and noise power $\sigma^{2}$,
the post-combining SNR is $(P/\sigma^{2})\,NM\,G$, where
\begin{equation}\label{eq:g}
G \triangleq \frac{1}{NM}\,\bigl|\wrx^{H}\bH\,\wt\bigr|^{2} \in [0,1]
\end{equation}
is the normalized array gain; $G = 1$ if and only if all $NM$ element
pairs combine coherently. All gains in this paper are power (SNR)
quantities.

We study the family of \emph{focused} beamformers: a pair that
conjugates the spherical phase toward axial depths $r_{\mathrm{t}}$ at
the transmitter and $r_{\mathrm{r}}$ at the receiver,
\begin{equation}\label{eq:weights}
w_{\mathrm{t},n} = \frac{1}{\sqrt N}\,
e^{\,jk\sqrt{r_{\mathrm{t}}^{2} + x_n^{2}}},
\qquad
w_{\mathrm{r},m} = \frac{1}{\sqrt M}\,
e^{-jk\sqrt{r_{\mathrm{r}}^{2} + y_m^{2}}} .
\end{equation}
The depths $(r_{\mathrm{t}}, r_{\mathrm{r}})$ are the only design
freedom. With
this choice, as shown below, each aperture's self-curvature is the part
the weights cancel. Two members of the family are singled out
throughout (Fig.~\ref{fig:geometry}): \emph{focusing}
($r_{\mathrm{t}} = r_{\mathrm{r}} = D$), each aperture matched to the
other's phase center (the ``focus on the user'' prescription of the
MIMO literature), and \emph{steering} ($r_{\mathrm{t}} =
r_{\mathrm{r}} = \infty$, i.e.\ a uniform phase ramp, the far-field
limit). They will prove to be the two extreme points of the family.
The optimum we derive in Section~\ref{sec:confocus}, \emph{con-focusing},
is the third reference point.

Besides the link Fresnel number $\NF$ and the aperture ratio $\rho$
of \eqref{eq:NF-def}, the analysis uses the per- and
combined-aperture numbers
\begin{equation}\label{eq:NF}
N_{\mathrm{t}} \triangleq \frac{L_{\mathrm{t}}^{2}}{\lambda D},
\qquad
N_{\mathrm{r}} \triangleq \frac{L_{\mathrm{r}}^{2}}{\lambda D},
\qquad
\nu_{\pm} \triangleq \frac{(L_{\mathrm{t}} \pm
L_{\mathrm{r}})^{2}}{4\lambda D} .
\end{equation}
The convention $\rho \le 1$ entails no loss of generality: the roles
of the two ends are symmetric. Here $N_{\mathrm{t}}$ and
$N_{\mathrm{r}}$ are the \emph{per-aperture} Fresnel numbers, measuring
how deeply each aperture, taken alone, operates in its own near field,
and $\NF = \sqrt{N_{\mathrm{t}} N_{\mathrm{r}}}$ is their geometric
mean. The quantities $\nu_{\pm}$ are the Fresnel numbers of the sum and
difference apertures $L_{\mathrm{t}} \pm L_{\mathrm{r}}$: as we shall
see, $\nu_{+}$ controls the quadratic phase accumulated by steering
across the joint aperture (and, later, the depth of focus of the
con-focusing law), while the identity $\nu_{+} - \nu_{-} =
[(L_{\mathrm{t}}+L_{\mathrm{r}})^{2} - (L_{\mathrm{t}}-L_{\mathrm{r}})^{2}]
/(4\lambda D) = \NF$ ties all the parameters together.

\subsection{The residual phase}

The Fresnel reduction turns the gain \eqref{eq:g} into a single
integral whose phase is the object of the whole analysis, and we derive
it here.

\begin{assumption}\label{ass:paraxial}
$L_{\mathrm{t}} + L_{\mathrm{r}} \ll D$, so that the Fresnel
(quadratic) expansion of the square roots applies.
\end{assumption}

Under Assumption~\ref{ass:paraxial} each square root in \eqref{eq:dist}
and \eqref{eq:weights} admits the expansion $\sqrt{a^{2} + \xi^{2}} =
a + \xi^{2}/(2a) - \xi^{4}/(8a^{3}) + \bigO(\xi^{6}/a^{5})$. Applied to
the path
\eqref{eq:dist},
\begin{equation}\label{eq:fresnel}
d_{nm} = D + \frac{x_n^{2}}{2D} + \frac{y_m^{2}}{2D}
- \frac{x_n y_m}{D} + \varepsilon_{nm},
\quad
|\varepsilon_{nm}| \le \frac{(L_{\mathrm{t}} + L_{\mathrm{r}})^{4}}
{128\,D^{3}} .
\end{equation}
A single $(m,n)$ term of the gain \eqref{eq:g} is
$w_{\mathrm{r},m}^{*} H_{mn} w_{\mathrm{t},n}$; by
\eqref{eq:channel}--\eqref{eq:weights} its phase is
$k\sqrt{r_{\mathrm{r}}^{2} + y_m^{2}} - k d_{nm}
+ k\sqrt{r_{\mathrm{t}}^{2} + x_n^{2}}$. Expanding the two weight
square roots to the same order as \eqref{eq:fresnel} and substituting,
\begin{equation}\label{eq:expanded}
k\Bigl(r_{\mathrm{t}} + \tfrac{x_n^{2}}{2r_{\mathrm{t}}}\Bigr)
+ k\Bigl(r_{\mathrm{r}} + \tfrac{y_m^{2}}{2r_{\mathrm{r}}}\Bigr)
- k\Bigl(D + \tfrac{x_n^{2}}{2D} + \tfrac{y_m^{2}}{2D}
- \tfrac{x_n y_m}{D}\Bigr).
\end{equation}
The constant $k(r_{\mathrm{t}} + r_{\mathrm{r}} - D)$ is independent of
$(m,n)$ and drops out of the modulus \eqref{eq:g}. Grouping the rest by
monomial leaves the residual phase
\begin{equation}\label{eq:psi}
\Psi_{mn} = \frac{k x_n^{2}}{2}\Bigl(\frac{1}{r_{\mathrm{t}}}
- \frac{1}{D}\Bigr)
+ \frac{k y_m^{2}}{2}\Bigl(\frac{1}{r_{\mathrm{r}}} - \frac{1}{D}\Bigr)
+ \frac{k x_n y_m}{D} .
\end{equation}
The expansion neglects, besides the path residual $\varepsilon_{nm}$ of
\eqref{eq:fresnel}, the quartic remainders $\bigO(kL_{\mathrm{t}}^{4}/
r_{\mathrm{t}}^{3})$ and $\bigO(kL_{\mathrm{r}}^{4}/r_{\mathrm{r}}^{3})$
of the two weight square roots; on the focused family
($r_{\mathrm{t}} \ge D/(1+\rho)$, $r_{\mathrm{r}} \ge \rho D/(1+\rho)$,
down to the con-focusing pair) all three are $\bigO(k(L_{\mathrm{t}} +
L_{\mathrm{r}})^{4}/D^{3}) \ll 1$ under Assumption~\ref{ass:paraxial}.

The first two terms are the residual self-curvatures of the two
apertures, left uncompensated when the focal depths differ from $D$;
the third is the bilinear coupling between the apertures, which no
choice of $(r_{\mathrm{t}}, r_{\mathrm{r}})$ can remove.

It remains to make \eqref{eq:psi} dimensionless. Set
$u = x_n/L_{\mathrm{t}}$ and $v = y_m/L_{\mathrm{r}}$, so that
$u, v \in [-\tfrac12, \tfrac12]$. Using $\tfrac{k}{2}L_{\mathrm{t}}^{2}
(\tfrac{1}{r_{\mathrm{t}}} - \tfrac{1}{D}) = \pi N_{\mathrm{t}}
(\tfrac{D}{r_{\mathrm{t}}} - 1)$, the analogous identity for
$L_{\mathrm{r}}$, and $k L_{\mathrm{t}} L_{\mathrm{r}}/D = 2\pi\NF$, the
residual phase becomes $\Psi_{mn} = \Phi(u,v)$ with the quadratic form
\begin{equation}\label{eq:Phi}
\Phi(u,v) = A u^{2} + B v^{2} + 2\pi\NF\,uv,
\end{equation}
in which the focal depths enter only through the two curvatures
\begin{equation}\label{eq:AB}
A \triangleq \pi N_{\mathrm{t}}\Bigl(\frac{D}{r_{\mathrm{t}}} - 1\Bigr),
\qquad
B \triangleq \pi N_{\mathrm{r}}\Bigl(\frac{D}{r_{\mathrm{r}}} - 1\Bigr).
\end{equation}
Each depth sets one curvature, independently of the other; the cross
term $2\pi\NF\,uv$, common to every pair, is the non-separable residual
no rank-one beamformer can touch.

\subsection{The field integral}

For uniformly spaced arrays the normalized sums in \eqref{eq:g}
converge, as $N, M \to \infty$, to integrals over $u, v$ of \eqref{eq:g}). The continuous-aperture gain is then
$G = |I|^{2}$ with the field integral
\begin{equation}\label{eq:gAB}
I = \iint_{[-1/2,\,1/2]^{2}} e^{\,j\Phi(u,v)}\,du\,dv ,
\end{equation}
the discrete-to-continuous gap being controlled in
Section~\ref{sec:gains}. Equations \eqref{eq:Phi}--\eqref{eq:gAB} are
the object of the analysis: a single integral, parameterized by the two
curvatures $(A,B)$ and the fixed cross term $2\pi\NF$.

The two named schemes are the extreme curvature points of
\eqref{eq:AB}: focusing, $r_{\mathrm{t}} = r_{\mathrm{r}}
= D$, gives $(A,B) = (0,0)$, where the self-curvatures vanish and only
the bilinear coupling survives; steering, $r_{\mathrm{t}} =
r_{\mathrm{r}} = \infty$, gives $(A,B) = (-\pi N_{\mathrm{t}}, -\pi
N_{\mathrm{r}})$, where $N_{\mathrm{t}} N_{\mathrm{r}} = \NF^{2}$ turns
\eqref{eq:Phi} into the perfect square $-\pi(\sqrt{N_{\mathrm{t}}}\,u -
\sqrt{N_{\mathrm{r}}}\,v)^{2}$. Section~\ref{sec:gains} characterizes
$G$ over the whole family.

\section{Closed-Form Gains and the Crossover}\label{sec:gains}

Section~\ref{sec:model} reduced the array gain to the field integral
\eqref{eq:gAB}, with the array gain as $G = |I|^{2}$. The family
of focused beamformers is thus the two-parameter set $(A,B)$, and we
characterize $G$ over it.

\subsection{The confocal law}

\begin{figure}[!tb]
\centering
\includegraphics[width=\columnwidth]{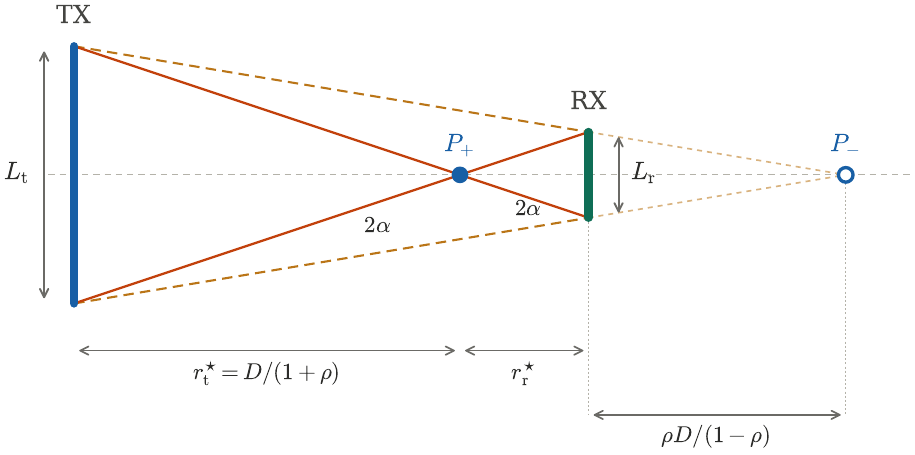}
\caption{The two solutions of the con-focusing law in one geometry
($\rho = 1/3$). The internal similitude center $P_{+}$ (the
con-focal point: both apertures subtend $2\alpha$) maps the
transmit aperture onto the receive aperture inverted (solid rays,
magnification $-\rho$); the external center $P_{-}$, at
$\rho D/(1-\rho)$ beyond the receiver, maps it erect (dashed rays,
$+\rho$). Either way the beam arrives with width exactly
$L_{\mathrm{r}}$; for $\rho = 1$, $P_{-}$ recedes to infinity and
the conjugate branch is plain steering. The gain along this family
of common foci is the profile of Fig.~\ref{fig:sweep}.}
\label{fig:law}
\end{figure}

The two curvatures act on the gain through a single combination of
them. Write the residual phase \eqref{eq:Phi} as the quadratic form
\begin{equation}\label{eq:quadform}
\Phi(u,v) = \bm{q}^{H}\Omega\,\bm{q}, \qquad \bm{q} = (u,v)^{\top},
\end{equation}
where $\Omega$ is the symmetric matrix
\begin{equation}\label{eq:Omega}
\Omega = \begin{pmatrix} A & \pi\NF\\[2pt] \pi\NF & B \end{pmatrix},
\qquad
\Delta \triangleq \det\Omega = AB - \pi^{2}\NF^{2} .
\end{equation}
We call the curve $\Delta = 0$, i.e.\ $AB = \pi^{2}\NF^{2}$, the
\emph{confocal locus}. On it $\Omega$ is singular: the residual phase
\eqref{eq:Phi} loses curvature along one direction and stays flat
across a whole strip of the joint aperture, the element pairs add
coherently, and the gain is as large as a single beam permits. Off it
the phase curves both ways, the integrand oscillates, and the gain
falls. How far it falls is set by how far $(A,B)$ lies from the
locus, measured by the one scalar $\Delta$, the signed distance, and
not by $A$ and $B$ separately. The name is geometric: on the locus the
two apertures share a common focus, an axial point from which the
converging transmit beam just fills the receiver
(Fig.~\ref{fig:law}); Section~\ref{sec:confocus} locates these common
foci, the centers of similitude of the pair, and shows the
unit-modulus law aimed at them is asymptotically optimal.

\begin{lemma}[Confocal law]\label{lem:delta}
Let $\mu_{1},\mu_{2}$ be the eigenvalues of $\Omega$, so that
$\mu_{1}\mu_{2} = \Delta$.
\begin{enumerate}
\item[(i)] \emph{(Decay.)} For $\min_{i}|\mu_{i}| \gg 1$,
\begin{equation}\label{eq:deltalaw}
G(A,B) = \Theta\!\left(\frac{1}{|\Delta|}\right),
\qquad \log G = -\log|\Delta| + \bigO(1),
\end{equation}
the implied constant tending to $\pi^{2}$ as $\min_{i}|\mu_{i}| \to
\infty$.
\item[(ii)] \emph{(Cap.)} As $\Delta \to 0$ the decay law ceases to
hold and $G$ is bounded by the rank-one eigenbound, of order
$\NF^{-1}$, attained in order on the locus by
Theorem~\ref{thm:closedform}.
\end{enumerate}
\end{lemma}

The lemma collapses the two-parameter design $(A,B)$ to a single aim:
get close to the confocal locus. Far from it, every decade of $|\Delta|$
costs a decade of gain; approaching it, $G$ rises but cannot exceed the
rank-one eigenbound, of order $\NF^{-1}$, the maximum a single beam can
collect. The hypothesis $\min_{i}|\mu_{i}| \gg 1$ marks
this decaying branch, where the phase oscillates many times along both
principal axes; it fails near the locus, where one eigenvalue vanishes
and the saturation takes over.

The phase \eqref{eq:Phi} has a
single stationary point, the aperture center, where it is flat;
everywhere else it oscillates and cancels. The integral localizes at
the center and, taken over the whole plane, equals
$|I| = \pi/\sqrt{|\Delta|}$. The finite aperture adds only an edge
correction of relative size $\bigO(\min_{i}|\mu_{i}|^{-1/2})$, so the
order in \eqref{eq:deltalaw} is exact while its constant is reached
slowly. On the locus the stationary point degenerates and the gain is
capped by the eigenbound. The full argument is in
Appendix~B of \suppl.

\subsection{Focusing and Steering}

Equation \eqref{eq:deltalaw} settles the comparison \emph{in order}:
the two schemes sit at opposite ends of it. Focusing, $(A,B)=(0,0)$,
has $\Delta = -\pi^{2}\NF^{2}$: it lies as far from the confocal locus
as the geometry allows, so $\gfoc = \Theta(\NF^{-2})$. Steering,
$(-\pi N_{\mathrm{t}},-\pi N_{\mathrm{r}})$, has $\Delta = 0$: it lies
\emph{on} the locus, so $\gff = \Theta(\NF^{-1})$. The near-field
penalty of focusing is geometric, not a matter of constants: it focuses
at the worst point of the family. The confocal locus is a curve,
not a point: steering is one member, and the member that maximizes the
gain is the order-optimal con-focusing pair of
Section~\ref{sec:confocus}. This fixes the deep-near-field ranking and
the ten-dB-per-decade gap, but not the crossover: the constant in
\eqref{eq:deltalaw} is reached slowly and $\NFs$ lies at $\NF =
\bigO(1)$, outside the asymptotic regime. The value $\NFs = 1.947$
follows from the exact closed forms (Theorem~\ref{thm:closedform}) and
the verified certificate (Theorem~\ref{thm:regions}(d)), not from the
asymptotic law.


\begin{theorem}[Exact gains of focusing and steering]\label{thm:closedform}
In the continuous-aperture limit,
\begin{align}
\gfoc(\NF) &= \Bigl[\frac{2}{\pi\NF}\,
\Si\!\Bigl(\frac{\pi\NF}{2}\Bigr)\Bigr]^{2},
\label{eq:gfoc-cf}\\
\gff &= \frac{1}{\NF^{2}}\,
\bigl|\nu_{+}\varphi(\nu_{+}) - \nu_{-}\varphi(\nu_{-})\bigr|^{2}.
\label{eq:gff-cf}
\end{align}
where
\begin{equation}\label{eq:phi}
\varphi(\nu) \triangleq \sqrt{\frac{2}{\nu}}\,F(\sqrt{2\nu})
- \frac{1 - e^{-j\pi\nu}}{j\pi\nu},
\qquad \varphi(0^{+}) = 1 .
\end{equation}
For equal apertures, $\nu_{-} = 0$ and $\gff = |\varphi(\NF)|^{2}$.
\end{theorem}

\begin{IEEEproof}
At $(A,B) = (0,0)$ the integrand of \eqref{eq:gAB} is $e^{\,j2\pi\NF
uv}$, linear in $v$ at fixed $u$. The inner integral is
\begin{equation}\label{eq:inner}
\int_{-1/2}^{1/2} e^{\,j2\pi\NF uv}\,dv = \frac{\sin(\pi\NF u)}
{\pi\NF u},
\end{equation}
and the substitution $t = \pi\NF u$ in the even outer integral gives
$I = (2/\pi\NF)\,\Si(\pi\NF/2)$, which is \eqref{eq:gfoc-cf}. At
$(A,B) = (-\pi N_{\mathrm{t}}, -\pi N_{\mathrm{r}})$ the form is a
perfect square,
\begin{equation}\label{eq:square}
\Phi(u,v)
= -\pi\bigl(\sqrt{N_{\mathrm{t}}}\,u - \sqrt{N_{\mathrm{r}}}\,v\bigr)^{2}
= -\frac{\pi s^{2}}{\lambda D},
\end{equation}
a function of $s = x - y$ alone. Reducing \eqref{eq:gAB} to the
physical apertures and integrating out the direction orthogonal to $s$
leaves
\begin{equation}\label{eq:overlap}
I = \frac{1}{L_{\mathrm{t}}L_{\mathrm{r}}}\int T(s)\,
e^{-j\pi s^{2}/(\lambda D)}\,ds,
\end{equation}
where $T = f_{c_{+}} - f_{c_{-}}$ is the cross-correlation of the two
aperture indicators, with triangle $f_{c}(s) = (c - |s|)_{+}$ and
$c_{\pm} = (L_{\mathrm{t}} \pm L_{\mathrm{r}})/2$. Each triangle obeys
\begin{equation}\label{eq:tri}
\int f_{c}(s)\,e^{-j\pi s^{2}/(\lambda D)}\,ds = c^{2}\varphi(\nu),
\qquad \nu = c^{2}/(\lambda D),
\end{equation}
which defines $\varphi$; with $c_{\pm}^{2} = \nu_{\pm}\lambda D$ and
$L_{\mathrm{t}}L_{\mathrm{r}} = \NF\lambda D$ this gives
$I = [\nu_{+}\varphi(\nu_{+}) - \nu_{-}\varphi(\nu_{-})]/\NF$, hence
\eqref{eq:gff-cf}. The limit $\varphi(0^{+}) = 1$ follows from
$C(z) = z + \bigO(z^{5})$, $S(z) = \bigO(z^{3})$.
\end{IEEEproof}

The regime analysis below rests on an exact identity for $\varphi$.
With $z = \sqrt{2\nu}$, so that $F'(z) = e^{-j\pi\nu}$, definition
\eqref{eq:phi} reads $\nu\varphi(\nu) = zF(z) - (1 - e^{-j\pi\nu})/(j\pi)$,
whence
\begin{equation}\label{eq:dnuphi}
\frac{d}{d\nu}\bigl[\nu\varphi(\nu)\bigr]
= \frac{F(\sqrt{2\nu})}{\sqrt{2\nu}},
\end{equation}
the oscillatory boundary terms cancelling identically; differentiating
the Fresnel transform of the triangle $f_{c}$ in its half-base
replaces it by the rectangle, whose transform is $F$. Integration by
parts in $F$ gives $|F(z) - (1-j)/2| \le 2/(\pi z)$
\cite[Ch.~3]{Olver1974}.

The structure of \eqref{eq:gfoc-cf} is worth noting. Focusing depends
only on the difference $\nu_{+} - \nu_{-} = \NF$, steering on
$\nu_{+}$ and $\nu_{-}$ separately. A strongly asymmetric link can
have $\nu_{\pm}$ large while $\NF$ stays small (the MISO-like
regime), whereas a symmetric link cannot; the dichotomy of this paper
follows from this asymmetry.

\subsection{Operating regimes and crossover}

\begin{theorem}[Operating regimes and crossover]\label{thm:regions}
The gains \eqref{eq:gfoc-cf} satisfy the following.
\begin{enumerate}
\item[(a)] As $\nu_{+} \to 0$, $\gfoc = 1 + \lilo(1)$ and
$\gff = 1 + \lilo(1)$.
\item[(b)] For $\NF \lesssim 1$ and $\nu_{0} \triangleq
L_{\mathrm{t}}^{2}/(4\lambda D) \gg 1$, $\gfoc = 1 + \bigO(\NF^{2})$ and
\begin{equation}\label{eq:misoNF}
\gff = \frac{|F(\sqrt{2\nu_0})|^{2}}{2\nu_0}
= \frac{1}{4\nu_{0}}\bigl(1 + \bigO(\nu_0^{-1/2})\bigr).
\end{equation}
\item[(c)] For $\NF \to \infty$ at fixed $\rho$,
\begin{equation}\label{eq:deepNF}
\gfoc = \frac{1}{\NF^{2}}\bigl(1 + \bigO(\NF^{-1})\bigr),
\quad
\gff = \frac{\rho}{\NF} + \bigO(\NF^{-3/2}).
\end{equation}
\item[(d)] For $\rho = 1$, $\gfoc - \gff$ has exactly one zero,
\begin{align}\label{eq:NFstar}
\NFs &= 1.947\ldots, \\ \notag
\gfoc(\NFs) &= \gff(\NFs) = 0.3662\ldots,
\end{align}
with $\gfoc > \gff$ for $\NF < \NFs$ and $\gfoc < \gff$ for
$\NF > \NFs$.
\end{enumerate}
\end{theorem}

The proof is given in Appendix~B of \suppl; we comment on
the four parts in turn. Part (a) is the joint far field: the two
schemes are equivalent to first order. The hypotheses of part (b)
force $\rho \ll 1$; this is the MISO-like regime, where focusing
wins by the power factor $4\nu_{0}$, and \eqref{eq:misoNF} is the
classical MISO result with its exact constant: the gain of steering
toward a point receiver at finite distance is the well-known Fresnel
defocusing factor of a single aperture. Part (c) is the deep near
field, the heart of the dichotomy: $\gff/\gfoc = \rho\NF\,(1+\lilo(1))$,
so steering wins with an SNR advantage of $+10$~dB per decade of
Fresnel number. The physical reading is clean.
The quadratic phase left uncompensated by steering is
coherent on the diagonal strip $|x - y| \lesssim \sqrt{\lambda D}$ of
the joint aperture, whose area is a fraction $\Theta(\NF^{-1/2})$ of
the whole; the bilinear kernel left by focusing is coherent only on a
hyperbolic neighborhood of the axes of effective area
$\Theta(\lambda D)$, a fraction $\Theta(\NF^{-1})$. The received
power is the square of the coherent fraction, and the squared ratio
of the two fractions is the $\NF$ power law. It is worth noting that
for $\rho < 1$ the deep-near-field steering gain in \eqref{eq:deepNF}
equals $\lambda D/L_{\mathrm{t}}^{2} = 1/N_{\mathrm{t}}$: it is set
by the \emph{transmit} aperture alone, because the smaller (receive)
aperture is fully coherent under the quadratic phase while the larger
(transmit) one is not.
Part (d) pins the crossover: below $\NFs$ focusing wins by under
$1$~dB, above it steering wins without bound.

\begin{remark}
\label{rem:geometry}
The value \eqref{eq:NFstar} is universal with respect to $N$, $M$,
$\lambda$, $D$, and the element spacing (Section~\ref{sec:gains},
sampling conditions), but it is specific to uniformly excited
\emph{linear} apertures with $\rho = 1$. For square uniform planar
arrays the problem factorizes per axis and the crossover occurs at
the same value of the \emph{per-axis} Fresnel number; for circular
uniform apertures it moves to a larger value (numerically $c^{\star}
\approx 3.81$), in line with the tapered-illumination threshold of
Soejima~\cite{Soejima1963};\footnote{Soejima's terminology is inverted
with respect to modern near-field usage: his ``focused'' denotes
constant phase over the aperture (steering here), and his ``defocused''
the spherical phase compensation that the XL-MIMO literature calls
focusing.} amplitude taper shifts it further. What is invariant is the structure: a unique
crossover governed by a Fresnel number, with the laws
\eqref{eq:deepNF}.
\end{remark}

\begin{corollary}[Asymmetry map]\label{cor:map}
Let $\sigma = (1+\rho^{2})/(2\rho)$ and
$\NFs(\rho) \triangleq \sup\{\NF : \gfoc \ge \gff\}$, the largest
Fresnel number at which focusing still matches steering and hence the
operational boundary, beyond which steering wins for every larger
$\NF$.
\begin{enumerate}
\item[(a)] For every $\rho$, focusing wins below an $\bigO(1/\sigma)$
Fresnel number and steering wins above an $\bigO(1/\rho)$ one (explicit
thresholds in Appendix~B of \suppl); hence $\NFs(\rho)$ is
finite.
\item[(b)] $\NFs(1) = 1.947$, with a unique crossing; for every
$\rho \in [0.1, 1]$ the crossings of $\gfoc - \gff$ are confined to an
explicit narrow window (Appendix~B of \suppl), unique except
near $\rho \approx 0.20$ and $\rho \approx 0.11$, where three to five
crossings occur.
\item[(c)] $\NFs(\rho) \to \infty$ as $\rho \to 0$.
\end{enumerate}
\end{corollary}

The thresholds in (a) are explicit in Appendix~B of \suppl;
part (b) rests on the verified certificate of
Theorem~\ref{thm:regions}(d), and its multiple crossings are the
lobe-jump kinks visible in the boundary map of
Section~\ref{sec:numerical}. Part (c) states that the focusing
advantage is an asymmetry phenomenon: the classical MISO result is
the $\rho \to 0$ boundary case of the map.

Two idealizations remain, continuous apertures and unit-modulus
entries; both are controlled. On a discrete array the two gains
are the finite sums
\begin{align}
\gfoc &= \frac{1}{N^{2}M^{2}}\Bigl|\sum_{n,m} e^{\,jk x_n y_m/D}
\Bigr|^{2},\label{eq:gfoc-sum}\\
\gff &= \frac{1}{N^{2}M^{2}}\Bigl|\sum_{n,m}
e^{-jk(x_n - y_m)^{2}/(2D)}\Bigr|^{2},\label{eq:gff-sum}
\end{align}
the finite-aperture forms of \eqref{eq:gfoc-cf}.

\begin{proposition}[Discretization and path loss]\label{prop:amplitude}
Let $d_{\mathrm{t}} = L_{\mathrm{t}}/(N-1)$, $d_{\mathrm{r}} =
L_{\mathrm{r}}/(M-1)$.
\begin{enumerate}
\item[(a)] The discrete sums \eqref{eq:gfoc-sum}--\eqref{eq:gff-sum}
agree with the continuous gains up to a relative
$\bigO(1/N) + \bigO(1/M)$ whenever each residual phase is sampled above
Nyquist, which for equal apertures ($\rho = 1$) is $N, M > 2\NF$
(Appendix~C of \suppl); below it, grating foci appear.
\item[(b)] Restoring the free-space amplitudes $H_{mn} =
(D/d_{nm})\,e^{-jkd_{nm}}$ perturbs every gain by a relative amount
bounded, uniformly over the $\NF$ sweep, by $C\eta(1 + \ln_{+}\NF)$
with $\eta \le (L_{\mathrm{t}}+L_{\mathrm{r}})^{2}/(8D^{2})$ and an
absolute constant $C$.
\end{enumerate}
\end{proposition}
Both are proved in Appendix~C of \suppl. Under the
paraxial Assumption~\ref{ass:paraxial} they are negligible: $\eta \le
(L_{\mathrm{t}}+L_{\mathrm{r}})^{2}/(8D^{2}) \ll 1$, so even the
focusing bound $\eta(1 + \ln_{+}\NF) \ll 1$ over any realistic Fresnel
number, and $N, M > 2\NF$ holds automatically for half-wavelength
arrays. The comparison is therefore genuinely between the two
unit-modulus, continuous-aperture gains \eqref{eq:gfoc-cf}.

\section{The Con-Focusing Law}\label{sec:confocus}

The focused-pair family \eqref{eq:gAB}, in the curvature coordinates
\eqref{eq:AB}, contains steering, focusing, and
everything between: $A = 0$ is focusing, $A = -\pi N_{\mathrm{t}}$
is steering ($r_{\mathrm{t}} = \infty$), and positive $A$ pulls the
focus closer than the receiver, with arbitrary transverse offsets
allowed. In these coordinates the landscape of the whole family is
completely characterized.

\begin{theorem}[The confocal landscape]\label{thm:landscape}
Consider focused pairs with fixed focal fractions
$r_{\mathrm{t}}/D$, $r_{\mathrm{r}}/D$ and fixed normalized offsets
as $\NF \to \infty$.
\begin{enumerate}
\item[(a)] $G = \Theta(\NF^{-1})$ if and only if $r_{\mathrm{t}} +
r_{\mathrm{r}} = D$ with signed depths, equivalently $AB =
\pi^{2}\NF^{2}$.
\item[(b)] On the set $AB = \pi^{2}\NF^{2}$, parametrized by
$A = \pi\NF\tau$, $B = \pi\NF/\tau$ with $\tau > 0$,
$\NF\,G \to \min(\tau, 1/\tau)$, maximized at $\tau = 1$, the
equal-subtense pair \eqref{eq:law}.
\item[(c)] Uniformly in the transverse offsets,
$G \le \bigl[4\sqrt{\pi}\,|AB - \pi^{2}\NF^{2}|^{-1/2}
+ \softO(\NF^{-1})\bigr]^{2}$; for
$r_{\mathrm{t}} = r_{\mathrm{r}} = D$,
$G = \bigO(\NF^{-2}\ln^{2}\NF)$.
\item[(d)] Steering and joint focusing at the link midpoint have
equal gain at every $(\NF, \rho)$.
\end{enumerate}
\end{theorem}

The proof is given in Appendix~D of \suppl. The condition in
(a) is the confocality condition of Boyd--Kogelnik resonator
theory~\cite{BoydKogelnik1962}: the two foci coincide at a common
point, possibly virtual; steering satisfies it with focus at
infinity. The parameter $\tau$ in (b) tracks the position of the
common focal point along the axis. The case $r_{\mathrm{t}} =
r_{\mathrm{r}} = D$ in (c) is focusing: each aperture
focuses on the other's location, two distinct points. The structural
reason for the failure of ``focus on the user'' is thus geometric:
the prescription is not confocal, so the joint phase admits no
stationary direction, and the
contributions of the two apertures can never be made to add along
any line of the joint aperture. Figure~\ref{fig:map} draws the map:
every scheme of the family is a point in the $(A,B)$ plane, all the
good ones sit on one hyperbola, and the prescription of the XL-MIMO
literature is the one popular point off it, with an offset
$\pi^{2}\NF^{2}$ that grows with the depth in the near field.

\begin{figure}[!tb]
\centering
\begin{tikzpicture}[scale=0.62, >=stealth]
\draw[->] (-3.5,0) -- (3.7,0)
  node[right,font=\scriptsize] {$A/\pi\NF$};
\draw[->] (0,-3.5) -- (0,3.7)
  node[above,font=\scriptsize] {$B/\pi\NF$};
\foreach \v in {1,-1} {
  \draw ($(\v,0)+(0,-0.07)$) -- ($(\v,0)+(0,0.07)$);
  \draw ($(0,\v)+(-0.07,0)$) -- ($(0,\v)+(0.07,0)$);
}
\node[font=\tiny,below] at (1,-0.09) {$1$};
\node[font=\tiny,below] at (-1,-0.09) {$-1$};
\draw[domain=0.30:3.3,smooth,very thick,green!55!black]
  plot (\x,{1/\x});
\draw[domain=-3.3:-0.30,smooth,very thick,green!55!black]
  plot (\x,{1/\x});
\node[font=\scriptsize,green!55!black,anchor=west] at (1.62,2.05)
  {confocal ridge $AB=\pi^{2}\NF^{2}$};
\fill[blue] (1,1) circle (2.6pt);
\node[font=\scriptsize,blue,anchor=south west] at (1.06,1.06)
  {$P_{+}$ (con-focal)};
\draw[blue,thick] (-1,-1) circle (2.6pt);
\node[font=\scriptsize,blue,anchor=north east] at (-1.06,-1.06)
  {$P_{-}$ (conjugate)};
\fill[black] (-3.07,-0.41) rectangle ++(0.15,0.15);
\node[font=\scriptsize,anchor=north] at (-2.95,-0.55)
  {FF--FF};
\node[red,font=\small] at (0,0) {$\boldsymbol\times$};
\node[font=\scriptsize,red,anchor=north west] at (-3,1.26)
  {focus on user};
\draw[->,gray] (0.16,0.16) -- (0.86,0.86);
\node[font=\scriptsize,gray,anchor=west] at (0.55,-1.15)
  {$G\!\sim\!|AB-\pi^{2}\NF^{2}|^{-1}$ off the ridge};
\end{tikzpicture}
\caption{The design map in the curvature coordinates \eqref{eq:AB},
drawn for the example geometry of Fig.~\ref{fig:law}
($\rho = 1/3$; axes in units of $\pi\NF$). Every scheme of the
focusing family is a point; all order-optimal pairs lie on the
confocal ridge (Theorem~\ref{thm:landscape}), where the two ends share a
focal point, real (upper branch) or virtual (lower branch); the
power-gain constant $\min(\tau,1/\tau)$ along it is maximized
at the similitude centers $P_{\pm}$, $A = B = \pm\pi\NF$
(Theorem~\ref{thm:confocus}). Steering is on the ridge, with
constant $\rho$; the con-focusing sweep of
Fig.~\ref{fig:sweep} walks along this hyperbola through $P_{-}$ and
$P_{+}$. ``Focus on the user'' is the origin: the one popular point
off the ridge.}
\label{fig:map}
\end{figure}
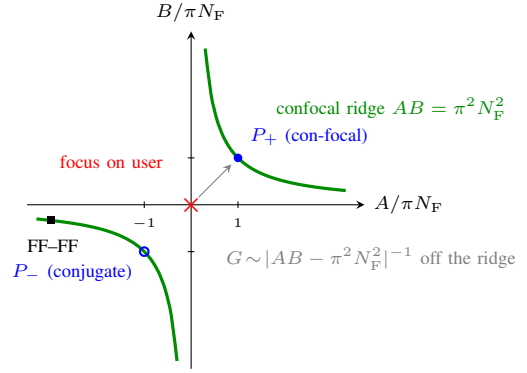

\begin{theorem}[Optimality of the con-focusing law]\label{thm:confocus}
Let the \emph{con-focusing pair} be the confocal pair focused at the
\emph{con-focal point}, the common point from which the two apertures
subtend equal angles,
\begin{equation}\label{eq:law}
r_{\mathrm{t}}^{\star} = \frac{D}{1+\rho},\qquad
r_{\mathrm{r}}^{\star} = \frac{\rho D}{1+\rho},
\end{equation}
i.e., normalized curvatures $A^{\star} = B^{\star} = \pi\NF$. For
$\NF \to \infty$ at fixed $\rho$:
\begin{enumerate}
\item[(i)] $\gcf = \NF^{-1}(1 + \bigO(\NF^{-1/2}\ln\NF))$;
\item[(ii)] every rank-one scheme, amplitude-tapered and
channel-aware included, satisfies
$G \le \geig = \NF^{-1}(1 + \bigO(\sqrt{\NF}\,e^{-\pi\NF}))$ for
continuous apertures;
\item[(iii)] $\gff/\gcf \to \rho$ and $\gfoc/\gcf \to \NF^{-1}$;
\item[(iv)] the conjugate pair $r_{\mathrm{t}} = D/(1-\rho)$,
$r_{\mathrm{r}} = -\rho D/(1-\rho)$ (curvatures
$A = B = -\pi\NF$) attains the same gain; for $\rho = 1$ it is steering.
\end{enumerate}
\end{theorem}

The proof is in Appendix~D of \suppl. Parts (i)--(ii) say the
con-focusing pair \eqref{eq:law} attains the rank-one eigenbound in both
order and leading constant: amplitude taper and channel knowledge are
asymptotically free, the relative margin being the
$\bigO(\NF^{-1/2}\ln\NF)$ of part~(i). Parts (i) and (ii) are
continuous-aperture statements; a discrete array adds the same
$\bigO(1/N) + \bigO(1/M)$ to both $\gcf$ and $\geig$
(Appendix~C of \suppl), leaving the ratio (and the
leading-constant-one conclusion) unchanged. Part (iii) places the
two fixed schemes against that bound: steering a constant factor
$\rho$ below it, $\gff/\gcf \to \rho$, and focusing a
further decade per decade of Fresnel number, $\gfoc/\gcf \to
\NF^{-1}$. Both are read off \eqref{eq:deepNF}. The conjugate pair
(iv) focuses both apertures on the external similitude center beyond
the receiver and attains the same gain; for $\rho = 1$ it is plain
steering. Remark~\ref{rem:optics} gives the geometric reading.

\begin{remark}\label{rem:optics}
The law \eqref{eq:law} has an elementary reading in geometric optics
(Fig.~\ref{fig:law}). An aperture of length $L_{\mathrm{t}}$ focused
at depth $r_{\mathrm{t}}$ produces, at range $D$, a beam of width
$W = L_{\mathrm{t}}\,|D/r_{\mathrm{t}} - 1|$. The condition
$W = L_{\mathrm{r}}$ (the beam fills the receive aperture
exactly) has two solutions, the two similitude centers of the
aperture pair: they are \eqref{eq:law} and its conjugate
(Theorem~\ref{thm:confocus}(iv)). The parameter $\tau$ of
Theorem~\ref{thm:landscape}(b) is the footprint ratio
$\tau = W/L_{\mathrm{r}}$, and the constant $\min(\tau, 1/\tau)$ is
a power count: for $W > L_{\mathrm{r}}$ the receiver intercepts only
the fraction $1/\tau$ of the beam; for $W < L_{\mathrm{r}}$ the beam
illuminates only the fraction $\tau$ of the receive elements. Steering has $W = L_{\mathrm{t}}$, hence $\tau = 1/\rho$: this is
the factor $\rho$ of Theorem~\ref{thm:regions}(c). Focusing shrinks
$W$ to the diffraction spot
$\lambda D/L_{\mathrm{t}} \ll L_{\mathrm{r}}$, and almost the entire
receive aperture sits idle. Leaving the confocal set costs fast
(Theorem~\ref{thm:landscape}(c)); moving along it costs only the
footprint mismatch. Figures~\ref{fig:law}, \ref{fig:map},
and~\ref{fig:sweep} show one example ($\rho = 1/3$, $\NF = 10$): in
Fig.~\ref{fig:sweep} the two peaks are the similitude centers, and
the valley between them is the user.
\end{remark}

\begin{figure}[!tb]
\centering
\includegraphics[width=\columnwidth]{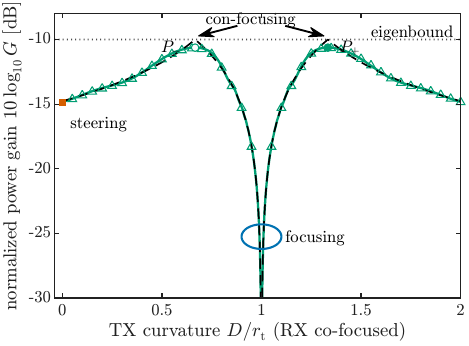}
\caption{Power gain along the con-focusing family ($\rho = 1/3$,
$\NF = 10$, the geometry of Fig.~\ref{fig:law}) as the common focus
moves along the axis, parametrized by the TX curvature
$D/r_{\mathrm{t}}$ (steering at $0$, the user at $1$). The similitude
centers $P_{\mp}$ at $1 \mp \rho$ are the two peaks (con-focusing),
riding the eigenbound; the user is the valley between them, where the
family collapses to the focusing null (ringed). Solid line: exact
one-dimensional ridge reduction (Appendix~D of \suppl); dashed:
crest envelope $\min(\tau,1/\tau)/\NF$; markers: exact spherical
channel, no Fresnel approximation.}
\label{fig:sweep}
\end{figure}

\begin{remark}\label{rem:semantics}
The optimal pair \eqref{eq:law} is itself a focused profile, so the
title deserves precision: what fails is the choice of focal point, not
the focusing operation. ``Focus on the user'' places the two foci on
the opposite arrays, $r_{\mathrm{t}} = r_{\mathrm{r}} = D$; they are
distinct, the pair is not confocal, and Theorem~\ref{thm:landscape}(c)
applies. The optimum instead places both foci on one shared point,
which is never the opposite array: plain steering for $\rho = 1$, a
mild curvature correction toward a virtual focus for $\rho < 1$, and
the user itself only as $\rho \to 0$, where it recovers the MISO
prescription. The near-field error of the XL-MIMO literature is thus
carrying the MISO focal point into a MIMO geometry.
\end{remark}

\begin{remark}\label{rem:pp}
For the LoS channel \eqref{eq:channel} the matrix $\bH$ is a
deterministic function of the link geometry, so channel knowledge
carries no information that geometry does not already contain;
Theorem~\ref{thm:confocus} makes the comparison quantitative. Every
rank-one scheme, channel-aware included, is bounded by $\geig$, and
the con-focusing pair attains $\geig$ with leading constant one from
three geometric quantities ($D$, $L_{\mathrm{t}}$,
$L_{\mathrm{r}}$): whatever full CSI can add vanishes as $\NF$
grows. Moreover, the standard way of spending CSI under the
unit-modulus constraint (keep the per-element phases of the
dominant singular vectors of $\bH$, set all amplitudes to
$1/\sqrt{N}$ and $1/\sqrt{M}$) falls below plain steering
beyond the crossover. The reason is that the singular
vectors carry a phase profile and an amplitude taper,
asymptotically the prolate taper of
Borgiotti~\cite{Borgiotti1966}; the projection keeps the phases and
discards the taper, and the taper is worth more. Both statements
are specific to the LoS model: in a multipath channel $\bH$ carries
information beyond geometry, and CSI recovers its usual role.
\end{remark}

In practice the law \eqref{eq:law} is built from a range estimate,
and its robustness is governed by the joint-aperture Fresnel number
$\nu_{+}$ of \eqref{eq:NF}: the larger $\nu_{+}$, the faster a range
error detunes the pair off the confocal set.

\begin{lemma}[Depth of focus]\label{lem:dof}
Let the con-focusing pair be built at the estimated range
$\widehat D = (1+\varepsilon)D$. There is an absolute constant
$c_{1} > 0$ such that the pair stays within $3$~dB of optimum for
$|\varepsilon| \le c_{1}/\nu_{+}$, and falls more than $3$~dB below
it for $|\varepsilon| \ge 16/(\pi\nu_{+})$. Equivalently, the
tolerated absolute ranging error is
\begin{equation}\label{eq:dof}
|\Delta D| \;=\; \Theta\!\left(\lambda
\Bigl(\frac{D}{L_{\mathrm{t}}+L_{\mathrm{r}}}\Bigr)^{2}\right),
\end{equation}
the classical depth of focus of an aperture of size
$L_{\mathrm{t}}+L_{\mathrm{r}}$ at f-number
$D/(L_{\mathrm{t}}+L_{\mathrm{r}})$.
\end{lemma}
The $3$-dB tolerance \eqref{eq:dof} is thus $\Theta(1/\nu_{+})$:
relative ranging accuracy must scale as the inverse joint-aperture
Fresnel number, while the absolute tolerance \emph{grows} with link
distance as $D^{2}$. Evaluating the closed-form gain at the perturbed ridge gives the sharp
constant $\varepsilon_{3\mathrm{dB}} \approx c_{0}/\nu_{+}$ with
$c_{0} \approx 3/\pi$, inside the analytic band derived in
Appendix~D of \suppl.

\begin{remark}[Which con-focal point: robustness under range drift]
\label{rem:robust}
The two similitude centers of Theorem~\ref{thm:confocus} share the
static optimum and the depth of focus of Lemma~\ref{lem:dof}, but their
\emph{off-optimum} landscapes differ. The focusing valley ($A = B = 0$)
lies between them in curvature: the internal center $P_{+}$
($A = B = \pi\NF$) sits on its high-curvature side, the external center
$P_{-}$ ($A = B = -\pi\NF$) on the steering side. A shrinking range
lowers both curvatures by the detuning (Appendix~D of \suppl), driving
$P_{+}$ toward $A = B = 0$: the pair enters the focusing valley, where
the off-ridge bound (Appendix~D of \suppl) forces the collapse. The same
detuning moves $P_{-}$ toward the steering point on the ridge, which is
itself order-optimal (Appendix~D of \suppl, at $\tau = 1/\rho$), so it
degrades only gently; a growing range is benign for both. The external
center thus has no catastrophic direction and is the robust operating
point under mobility, which is why the acquisition of
Section~\ref{sec:acquisition} targets it.
\end{remark}

\subsection{Beyond Boresight: Tilt and Roll}
\label{ssec:rotations}

The boresight assumption can be relaxed at no conceptual cost.
Relative to the line of sight (LoS), an array's orientation is set by
three rotations (Fig.~\ref{fig:angles}): two \emph{tilts} (pitch and
yaw, about the two transverse axes) and a \emph{roll} about the LoS
axis itself. The two tilts shorten the array's projection onto the
broadside plane, replacing the apertures by
$L_{\mathrm{t}}\cos\theta_{\mathrm{t}}$ and $L_{\mathrm{r}}
\cos\theta_{\mathrm{r}}$; as Corollary~\ref{cor:rotations} shows, that
projection is all that matters, and the boresight theory carries over
verbatim. The roll rotates the array within the transverse plane. For
a planar aperture it is benign (Remark~\ref{rem:planar}); for a
\emph{linear} array it sets the angle $\phi$ between the two array
axes, and that angle reshapes the focusing--steering comparison
(Proposition~\ref{prop:roll}). We take the tilts first, then the
roll.

\begin{figure}[!tb]
\centering
\subfloat[3D setup]{\includegraphics[height=1.1in]{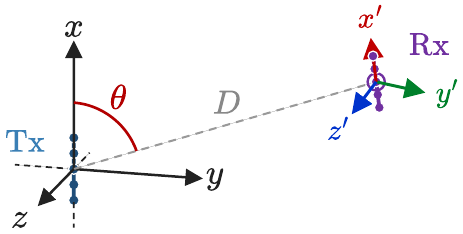}}\hfil
\subfloat[array rotations]{\includegraphics[height=1.1in]{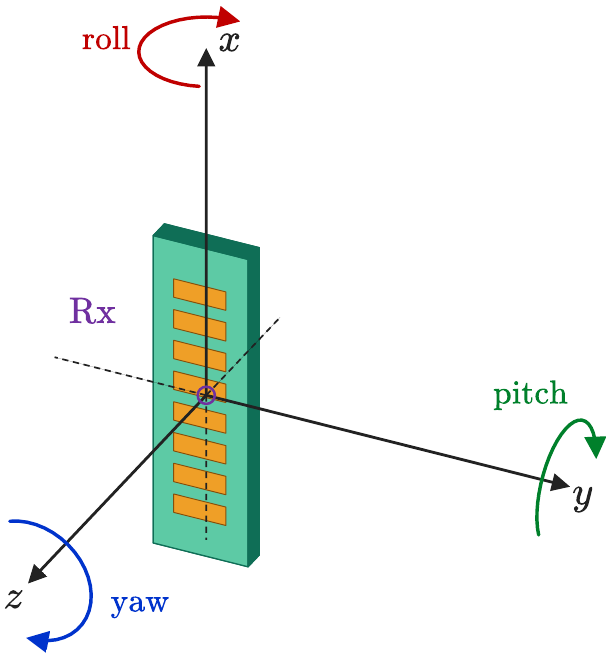}}
\caption{Off-boresight geometry. (a)~The transmit and receive uniform
linear arrays at separation $D$, with the global frame $(x,y,z)$ and
the rotated receive frame $(x',y',z')$; $\theta$ is the array angle to
the link axis. (b)~The three rigid rotations of the receive array
(roll, pitch, yaw). The two tilts about axes transverse to the line of
sight only rescale the broadside projections
$L_{\mathrm{t}}\cos\theta_{\mathrm{t}}$,
$L_{\mathrm{r}}\cos\theta_{\mathrm{r}}$ that enter the gains, so the
boresight theory applies verbatim under the map \eqref{eq:angmap};
the rotation about the line of sight turns the receive axis by $\phi$
relative to the transmit axis and, for linear arrays, weights the
cross-coupling by $\cos\phi$, which knocks steering off the confocal
ridge while the con-focusing law follows the geometry (benign for a
planar aperture, Remark~\ref{rem:planar}).}
\label{fig:angles}
\end{figure}

Let $\hat{\mathbf{z}}$ be the link axis
(transmit to receive centre) and $P = I - \hat{\mathbf{z}}
\hat{\mathbf{z}}^{\top}$ the projector onto the broadside plane. The
transmit array axis is a unit vector $\hat{\mathbf{a}}_{\mathrm{t}}$,
with tilt $\theta_{\mathrm{t}}$ defined by $\cos\theta_{\mathrm{t}} =
\|P\hat{\mathbf{a}}_{\mathrm{t}}\|$ and broadside projection
$\hat{\mathbf{e}}_{\mathrm{t}} = P\hat{\mathbf{a}}_{\mathrm{t}}/
\cos\theta_{\mathrm{t}}$; likewise $\hat{\mathbf{a}}_{\mathrm{r}},
\theta_{\mathrm{r}}, \hat{\mathbf{e}}_{\mathrm{r}}$ at the receiver. The
roll is the angle between the projected axes, $\cos\phi =
\hat{\mathbf{e}}_{\mathrm{t}}\cdot\hat{\mathbf{e}}_{\mathrm{r}}$. A
transmit element at array coordinate $x$ sits at
$x\hat{\mathbf{a}}_{\mathrm{t}}$ (longitudinal component
$x\sin\theta_{\mathrm{t}}$ along $\hat{\mathbf{z}}$ and transverse
offset $\tilde x\hat{\mathbf{e}}_{\mathrm{t}}$ with $\tilde x =
x\cos\theta_{\mathrm{t}}$), and likewise a receive element at $y$ has
longitudinal $y\sin\theta_{\mathrm{r}}$ and transverse
$\tilde y\hat{\mathbf{e}}_{\mathrm{r}}$, $\tilde y =
y\cos\theta_{\mathrm{r}}$. As $\hat{\mathbf{z}}$ is orthogonal to the
broadside plane, the element separation splits as $d^{2} = (D+a)^{2} +
c^{2}$, with longitudinal offset $a = y\sin\theta_{\mathrm{r}} -
x\sin\theta_{\mathrm{t}}$ and transverse part
\begin{equation}\label{eq:lawcos}
c^{2} = \|\tilde x\hat{\mathbf{e}}_{\mathrm{t}} -
\tilde y\hat{\mathbf{e}}_{\mathrm{r}}\|^{2}
= \tilde x^{2} + \tilde y^{2} - 2\,\tilde x\,\tilde y\cos\phi,
\end{equation}
the only place the roll enters. The paraxial expansion $d = D + a +
c^{2}/(2D) + \varepsilon_{3}$ holds under the cubic (coma-type)
condition quantified in Appendix~E of \suppl.
Steering conjugates the linear term $a$; focusing
conjugates in addition the separable curvatures $\tilde x^{2}$,
$\tilde y^{2}$. Two residual kernels survive in the projected
coordinates,
\begin{equation}\label{eq:twistkernels}
\underbrace{e^{\,jk\cos\phi\,\tilde x\tilde y/D}}_{\text{focusing}},
\qquad
\underbrace{e^{-jk(\tilde x^{2}+\tilde y^{2}
-2\cos\phi\,\tilde x\tilde y)/(2D)}}_{\text{steering}},
\end{equation}
which are exactly the boresight kernels \eqref{eq:gfoc-sum}--%
\eqref{eq:gff-sum} with two changes: the apertures are the projections
$L_{\mathrm{t}}\cos\theta_{\mathrm{t}}$, $L_{\mathrm{r}}
\cos\theta_{\mathrm{r}}$, and every cross term carries the factor
$\cos\phi$. Writing $\NFe = \NF\cos\theta_{\mathrm{t}}
\cos\theta_{\mathrm{r}}$, the two cases $\phi = 0$ and $\phi \neq 0$
separate cleanly.

\begin{corollary}[Covariance under coplanar tilt]\label{cor:rotations}
Let $\phi = 0$. Then \eqref{eq:twistkernels} are the boresight kernels
verbatim in the projected coordinates, so
Theorems~\ref{thm:closedform}--\ref{thm:confocus} and
Corollary~\ref{cor:map} hold unchanged under
\begin{equation}\label{eq:angmap}
(\NF, \rho) \;\longmapsto\;
\Bigl(\NFe,\;
\rho\,\tfrac{\cos\theta_{\mathrm{r}}}{\cos\theta_{\mathrm{t}}}\Bigr),
\qquad
\NFe = \NF\cos\theta_{\mathrm{t}}\cos\theta_{\mathrm{r}}.
\end{equation}
In particular the crossover moves to $\NFs(\theta) =
1.947/(\cos\theta_{\mathrm{t}}\cos\theta_{\mathrm{r}})$: tilt only
rescales the geometry, and the crossover, the asymptotic laws, and the
con-focusing law of this section carry over unchanged.
\end{corollary}

When the arrays are rolled out of a common plane the $\cos\phi$ in
\eqref{eq:twistkernels} multiplies the cross-coupling but not the self
terms, and this single asymmetry reshapes the comparison. The effect
is specific to the \emph{linear} apertures of our model, for which the
roll sets the angle $\phi$ between the two array lines; as
Remark~\ref{rem:planar} shows, it has no counterpart for planar
arrays.

\begin{proposition}[Roll]\label{prop:roll}
Fix $\theta_{\mathrm{t}}, \theta_{\mathrm{r}}$, write
$\NFe = \NF\cos\theta_{\mathrm{t}}\cos\theta_{\mathrm{r}}$, and let the
roll angle $\phi \in (0, \pi/2]$. To Fresnel order and within the
validity of Appendix~E of \suppl, the gains of the rolled link
satisfy the following.
\begin{enumerate}
\item[(a)] Steering leaves the confocal ridge of the modified kernel
by the discriminant $\pi^{2}(\NFe)^{2}\sin^{2}\phi$, and
$\gff = \softO(\NF^{-2})$.
\item[(b)] $\gfoc = \bigl[\tfrac{2}{\pi\NFe\cos\phi}\,
\Si(\tfrac{\pi\NFe\cos\phi}{2})\bigr]^{2}$ is increasing in $\phi$,
with $\gfoc \to 1$ as $\phi \to \pi/2$; hence $\gfoc - \gff$ has no
zero, and $\gfoc > \gff$ for all $\phi$ beyond a threshold $\phi_{c}$,
with $\phi_{c} \to 45^{\circ}$ as $\NFe \to \infty$.
\item[(c)] The ridge $AB = \pi^{2}(\NFe\cos\phi)^{2}$ is nonempty; its
equal-subtense pair, the symmetric member being
$r_{\mathrm{t}}^{\star} = r_{\mathrm{r}}^{\star} = D/(1+\cos\phi)$,
attains the order-optimal $G = \Theta((\NFe\cos\phi)^{-1})$ for
$\NFe\cos\phi \gg 1$; as $\phi \to \pi/2$ the rolled kernel becomes
separable and $G = \Theta(1)$.
\end{enumerate}
\end{proposition}

The proof is in Appendix~E of \suppl; we comment on the three
parts in turn.

Part~(a). Steering leaves an uncompensated quadratic phase
across the joint aperture. From the steering kernel in
\eqref{eq:twistkernels} this residual phase is
\begin{equation}\label{eq:twistphase}
\psi(\tilde x,\tilde y) = -\frac{k}{2D}\bigl(\tilde x^{2} + \tilde y^{2}
- 2\cos\phi\,\tilde x\tilde y\bigr).
\end{equation}
For parallel arrays ($\phi = 0$) it reduces to the perfect square
$-\frac{k}{2D}(\tilde x - \tilde y)^{2}$, constant along the diagonal
$\tilde x = \tilde y$: the two apertures stay correlated along that
strip and add coherently. Rolling the receive array replaces the unit
cross-coupling by $\cos\phi < 1$, so $\psi$ is no longer a perfect
square; it decorrelates across the aperture and the power transfer
becomes less efficient. Quantitatively, steering sits at curvatures
$A = -\pi N_{\mathrm{t}}^{\mathrm{eff}}$, $B = -\pi
N_{\mathrm{r}}^{\mathrm{eff}}$ in \eqref{eq:AB}, a distance
$AB - \pi^{2}(\NFe\cos\phi)^{2} = \pi^{2}(\NFe)^{2}\sin^{2}\phi$ off the
confocal ridge (Theorem~\ref{thm:landscape} with $\NF \mapsto
\NFe\cos\phi$), and the off-ridge bound (Appendix~D of \suppl) gives
$\gff = \softO(\NF^{-2})$, down from the $\Theta(\NF^{-1})$ of the
parallel case.

Part~(b). Focusing has already cancelled the self-terms
$\tilde x^{2}, \tilde y^{2}$; its residual is the cross-coupling alone,
$e^{jk\cos\phi\,\tilde x\tilde y/D}$ in \eqref{eq:twistkernels}, and the
roll enters only through $\cos\phi$. Its gain \eqref{eq:gfoc-cf} does
not fall with $\phi$: focusing is essentially insensitive to the roll.
Since steering decorrelates while focusing does not, the comparison
tips to focusing: the boresight crossover no longer occurs, and
focusing dominates steering beyond a threshold $\phi_{c}$ that tends to
$45^{\circ}$ as $\NFe \to \infty$ ($50^{\circ}$--$55^{\circ}$ at the
Fresnel numbers of Section~\ref{sec:numerical}).

Part~(c) makes precise why (b) holds. Order-optimal
$\Theta((\NFe\cos\phi)^{-1})$ gain is attained on the confocal ridge
$AB = \pi^{2}(\NFe\cos\phi)^{2}$; its equal-curvature point, the
con-focusing pair, has symmetric focal depth
$r_{\mathrm{t}}^{\star} = D/(1+\cos\phi)$, obtained from $A = B =
\pi\NFe\cos\phi$ in \eqref{eq:AB} at $\rho = 1$. As the roll grows the
focal point migrates from the midpoint $D/2$ at $\phi = 0$ toward the
receiver, reaching $r_{\mathrm{t}}^{\star} = D$ at $\phi = \pi/2$: the
optimal focus lands on the receive array itself, the con-focusing law
coincides with focusing, and the link has degenerated to a
single transmit-to-point mode (MISO-like). 

\begin{remark}\label{rem:planar}
Which rotation matters depends on the array's dimension. For a
\emph{linear} array the roll is decisive: a line has no extent
transverse to its axis, so the coupling reduces to the scalar product
$\tilde x\tilde y\cos\phi$ and vanishes when the two lines are
orthogonal, $\hat{\mathbf{e}}_{\mathrm{t}}\cdot\hat{\mathbf{e}}_
{\mathrm{r}}=0$, the collapse of Proposition~\ref{prop:roll}. For a
\emph{planar} array the roll is harmless: the transverse positions are
vectors $\mathbf{p},\mathbf{q}\in\mathbb{R}^{2}$, the coupling is the
inner product $\mathbf{p}\cdot\mathbf{q}$, and an in-plane roll
$\mathbf{q}\mapsto R_{\phi}\mathbf{q}$ leaves the focusing gain
$|\langle e^{jk\,\mathbf{p}\cdot R_{\phi}\mathbf{q}/D}\rangle|^{2}$
invariant for a disk and only weakly $\phi$-dependent for a
square, bounded by the aperture's modest anisotropy and exactly
invariant at multiples of $\pi/2$; in both, the
$\mathrm{DoF}\approx A_{\mathrm{t}}A_{\mathrm{r}}/
(\lambda D)^{2}$~\cite{Miller2000,LuZeng2024} and the whole
focusing--steering comparison are essentially unchanged. What degrades a
planar link instead is the \emph{tilt}: pitch and yaw foreshorten each
aperture by $\cos\theta$, which is exactly the covariance of
Corollary~\ref{cor:rotations}. In short, roll is the linear-array
effect, tilt is the planar-array effect.
\end{remark}

\section{Geometry-Free Acquisition of the Optimum}\label{sec:acquisition}

The con-focusing law \eqref{eq:law} returns the optimal depths from $D$,
$L_{\mathrm{t}}$, $L_{\mathrm{r}}$. The apertures are known from device
specifications; the range $D$ is not. Once standard sweeping has aligned
the far-field beams, the terminals exchange no geometry: only
synchronized probe indices and scalar received-power feedback, the
signaling of one beam-refinement round. We estimate $D$ from a few such
probes and apply \eqref{eq:law}.

\subsection{Observation model}\label{ssec:obs}

A probe is a focused pair aimed at $(r_{\mathrm{t}}, r_{\mathrm{r}})$
and excited by a unit-modulus pilot $s_{\ell}$, $|s_{\ell}| = 1$, of $L$
symbols,
\begin{equation}\label{eq:rx}
y_{\ell} = \sqrt{p}\,\wrx^{H}\bH\,\wt\,s_{\ell} + n_{\ell}, \qquad
n_{\ell} \sim \mathcal{CN}(0,\sigma^{2}) .
\end{equation}
Matched filtering and subtracting the calibrated noise floor give the
unbiased gain estimate
\begin{equation}\label{eq:Ghat}
\widehat G = \frac{1}{L\gamma}\sum_{\ell=1}^{L}
\frac{|s_{\ell}^{*}y_{\ell}|^{2}}{\sigma^{2}} - \frac{1}{\gamma},
\qquad \gamma \triangleq \frac{p\,NM}{\sigma^{2}},
\end{equation}
with $\gamma$ the per-probe SNR; its variance is that of a noncentral
power measurement,
\begin{equation}\label{eq:varG}
\operatorname{Var}\widehat G = \frac{2\gamma G + 1}{L\,\gamma^{2}} .
\end{equation}
Since $\rho, L_{\mathrm{t}}, L_{\mathrm{r}}$ are known, the gain
$G(r_{\mathrm{t}}, r_{\mathrm{r}}; D)$ of \eqref{eq:g} depends on the
lone unknown $D$ through $\NF = \rho\, d_{\mathrm{F}}/(2D)$
and the curvatures \eqref{eq:AB}. Acquisition is the estimation of $D$
from probes \eqref{eq:rx}, followed by \eqref{eq:law}.

\begin{algorithm}[!b]
\caption{Geometry-free acquisition of the confocal optimum}
\label{alg:acq}
\begin{algorithmic}[1]
\Require apertures $L_{\mathrm{t}}, L_{\mathrm{r}}$, ratio $\rho$,
wavelength $\lambda$
\Ensure focal depths $(r_{\mathrm{t}}^{\star}, r_{\mathrm{r}}^{\star})$
\State $c_{0} \gets (\infty, \infty)$,\quad
$c_{1} \gets (\infty,\, L_{\mathrm{r}}^{2}/\lambda)$
\Comment{far field; near-field probe}
\State $\mathcal{M} \gets \{(c_{0}, \widehat G(c_{0})),\,
(c_{1}, \widehat G(c_{1}))\}$
\State $\widehat D \gets \arg\max_{D}
\sum_{(c,g) \in \mathcal{M}} \ln p(g \mid D)$
\Comment{coarse ML estimate}
\Repeat
\State $c \gets \arg\max_{c}\,
[\partial_{D}G(c)]^{2}/\operatorname{Var}\widehat G\big|_{\widehat D}$
\Comment{Fisher-optimal, near $\widehat D/(1-\rho)$; \eqref{eq:fisher}}
\State $\mathcal{M} \gets \mathcal{M} \cup \{(c, \widehat G(c))\}$
\Comment{measure (over subcarriers)}
\State $\widehat D \gets \arg\max_{D}
\sum_{(c,g) \in \mathcal{M}} \ln p(g \mid D)$
\Comment{ML, $p$ from \eqref{eq:varG}}
\Until{$\sqrt{1/J(\widehat D)}\,/\widehat D < c_{1}/\nu_{+}$}
\Comment{depth of focus, Lemma~\ref{lem:dof}}
\State $r_{\mathrm{t}}^{\star} \gets \widehat D/(1-\rho)$,\quad
$r_{\mathrm{r}}^{\star} \gets -\rho\,r_{\mathrm{t}}^{\star}$
\Comment{external con-focal point $P_{-}$ (Remark~\ref{rem:robust})}
\State \Return $(r_{\mathrm{t}}^{\star}, r_{\mathrm{r}}^{\star})$
\end{algorithmic}
\end{algorithm}

\subsection{Solvability}\label{ssec:lever}

The gain $G(\cdot\,;D)$ is a known function of the lone unknown $D$, so
two probes at distinct configurations determine $D$ by inverting the
model; in the noiseless case the recovery is exact, for any $\NF$ and
$\rho$. Solvability is therefore not in question; only accuracy is.

The information sits where $G$ depends most sharply on $D$. Since $D$
enters only through $\NF = \rho\, d_{\mathrm{F}}/(2D)$ and
the curvatures \eqref{eq:AB}, and the far-field gain
$\gff \simeq \rho/\NF \propto D$ is the largest gain the link affords, a
far-field reference together with one near-field probe already pin the
range; sharper sensitivity, used below, comes from probing near the
confocal maximum, where a small change in $D$ shifts a needle-sharp
peak. A power ascent, by contrast, is useless: along the confocal
family $G$ is bimodal: steering and the con-focusing pair are
distinct maxima flanking the focusing valley. The inversion
must therefore use the model, not a search.

\subsection{Active refinement and probe placement}\label{ssec:crb}

At finite SNR the error is set by the variance \eqref{eq:varG}. For $L \gg 1$ the estimate $\widehat G$ is approximately
Gaussian (by central limit theorem), and $K$ probes at configurations $c_{k} = (r_{\mathrm{t}},
r_{\mathrm{r}})_{k}$ carry Fisher information
\begin{equation}\label{eq:fim}
J(D) = \sum_{k=1}^{K}
\frac{\bigl[\partial_{D}G(c_{k})\bigr]^{2}}
{\operatorname{Var}\widehat G_{k}}
= L\gamma^{2}\sum_{k=1}^{K}
\frac{\bigl[\partial_{D}G(c_{k})\bigr]^{2}}{2\gamma G(c_{k})+1},
\end{equation}
where $\partial_{D}G$ follows from the closed forms of
Section~\ref{sec:gains}, and the variance-derivative term is smaller by
a factor $\bigO(1/L)$ than the retained mean-gradient term. In the
asymptotic regime (large $L$, with the likelihood unimodal about the
truth) the maximum-likelihood estimator attains
$\operatorname{Var}\widehat D \to 1/J(D)$; the subcarrier averaging that
flattens the ripple is what keeps the likelihood unimodal and the
estimator in this regime. The design variable is then the choice of
probes: the contribution
$[\partial_{D}G]^{2}/\operatorname{Var}\widehat G$ vanishes where $G$ is
flat in $D$ (far field, deep focus) and is largest where $G$ varies
fastest, near the sharp confocal maximum. Selecting
\begin{equation}\label{eq:fisher}
c_{k+1} = \arg\max_{c}\,
\frac{\bigl[\partial_{D}G(c)\bigr]^{2}}{\operatorname{Var}\widehat G}
\bigg|_{\widehat D}
\end{equation}
is sequential $D$-optimal design: since $J$ is additive, each probe adds
the largest term available, minimizing $1/J(D)$. The maximizer lies on
the confocal ridge: off it the gain is small and slowly varying,
$G = \softO(\NF^{-2})$ with $\partial_{D}G$ of the same order
(Appendix~D of \suppl), whereas on the ridge the needle-sharp peak gives
$\partial_{D}G = \Theta(\nu_{+}G)$, so \eqref{eq:fisher} is dominated by
the on-ridge configuration on the steep flank of $P_{-}$.

Let $\widehat G(c)$ denote the gain estimate \eqref{eq:Ghat} returned by
a probe at configuration $c = (r_{\mathrm{t}}, r_{\mathrm{r}})$,
averaged over the pilot and a few subcarriers, whose ripple patterns
decorrelate and keep the likelihood off the sidelobes.
Algorithm~\ref{alg:acq} acquires the robust external con-focal point
$P_{-}$ (Remark~\ref{rem:robust}) by maximum likelihood. A
far-field reference and one near-field probe give a coarse $\widehat D$;
each further probe is then placed by \eqref{eq:fisher} (near the
predicted depth $\widehat r_{\mathrm{t}}^{\star} = \widehat D/(1-\rho)$,
where the gain is steepest in $D$), and $\widehat D$ is refreshed by
maximum likelihood over all probes, until the relative range
uncertainty $\sqrt{1/J(\widehat D)}/\widehat D$ falls within the
depth-of-focus tolerance $c_{1}/\nu_{+}$ of Lemma~\ref{lem:dof}. The
protocol is one beam-refinement round, with no ranging and no geometry
exchange; its accuracy is validated against the Cram\'er--Rao bound in
Section~\ref{sec:numerical}.

\section{Numerical Results}\label{sec:numerical}

\begin{figure}[!b]
\centering
\includegraphics[width=0.8\columnwidth]{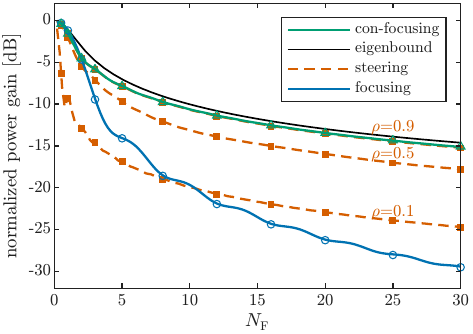}
\caption{Normalized power gain versus $\NF$. con-focusing attains the
rank-one eigenbound for any $\rho$; focusing decays as $\NF^{-2}$;
steering depends on the aperture ratio, shown for $\rho = 0.9, 0.5,
0.1$ (it coincides with con-focusing at $\rho = 1$). Lines: analysis;
markers: exact channel. The abscissa is the link Fresnel number
$\NF = \rho\,d_{\mathrm{F}}/(2D)$, with $d_{\mathrm{F}} =
2L_{\mathrm{t}}^{2}/\lambda$ the Fraunhofer distance of the larger
aperture.}
\label{fig:gains}
\end{figure}

\begin{figure}[!t]
\centering
\includegraphics[width=0.99\columnwidth]{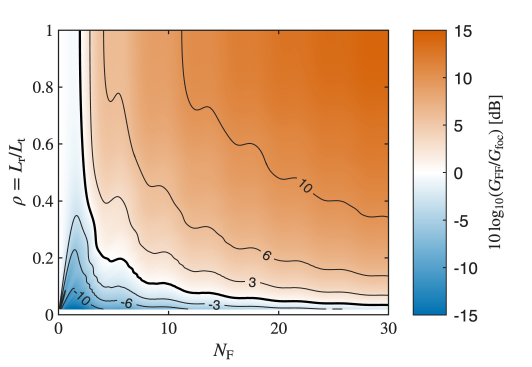}
\caption{Steering-to-focusing power-gain ratio
$10\log_{10}(\gff/\gfoc)$ [dB] over the $(\NF,\rho)$ plane. Vermillion:
steering wins; blue: focusing wins; the heavy $0$~dB contour is the
crossover separatrix $\NFs(\rho)$, tightest at $\rho = 1$ and receding
as $\rho \to 0$ (the MISO edge). The contour ripples near
$\rho \approx 0.2$ and $0.11$ are the multi-crossing windows of
Corollary~\ref{cor:map}.}
\label{fig:boundary}
\end{figure}

\begin{figure}[!b]
\centering
\includegraphics[width=0.8\columnwidth]{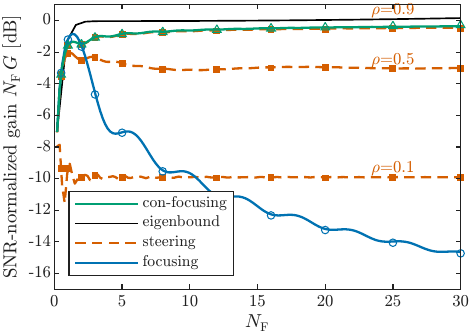}
\caption{Approach to the rank-one eigenbound. As $\NF$ grows the
con-focusing pair attains the bound in leading constant
($\NF G_{\mathrm{cf}} \to 0$~dB) and steering saturates a factor $\rho$
below, shown for $\rho = 0.9, 0.5, 0.1$; focusing decays as $\NF^{-1}$.
Lines: analysis; markers: exact channel.}
\label{fig:eig}
\end{figure}

\begin{figure*}[!t]
\centering
\subfloat[$\rho = 0.9$]{\includegraphics[width=0.32\textwidth]{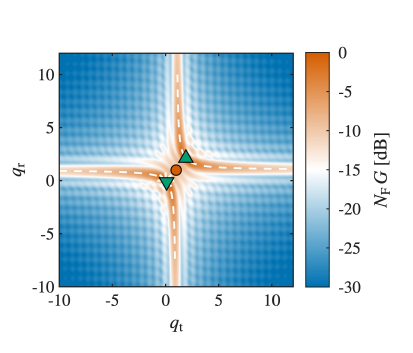}}\hfil
\subfloat[$\rho = 0.5$]{\includegraphics[width=0.32\textwidth]{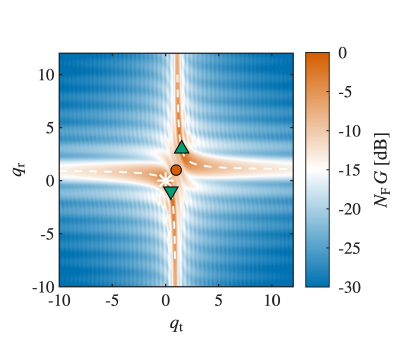}}\hfil
\subfloat[$\rho = 0.1$]{\includegraphics[width=0.32\textwidth]{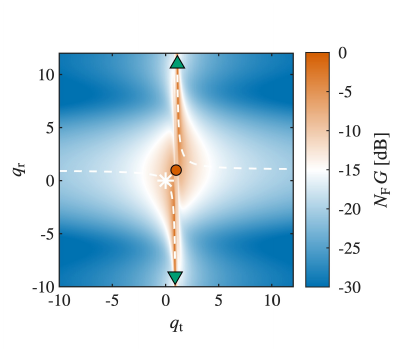}}
\caption{Normalized power gain $\NF G$ (eigenbound at $0$~dB) over the
curvature plane $(q_{\mathrm{t}}, q_{\mathrm{r}}) =
(D/r_{\mathrm{t}}, D/r_{\mathrm{r}})$ at $\NF = 10$, on common axes, for
(a)~$\rho = 0.9$, (b)~$\rho = 0.5$, (c)~$\rho = 0.1$. The confocal ridge
$(q_{\mathrm{t}}-1)(q_{\mathrm{r}}-1) = 1$ (dashed) is the same for
every $\rho$; focusing $(1,1)$ (circle) is a trough, steering $(0,0)$
(star) lies on the ridge, and the two con-focal centres
$P_\pm = (1\pm\rho,\,\pm(1\pm\rho)/\rho)$ (triangles) are the optima,
sliding outward as $\rho$ shrinks toward the MISO limit.}
\label{fig:landscape}
\end{figure*}

\begin{figure}[!t]
\centering
\includegraphics[width=0.8\columnwidth]{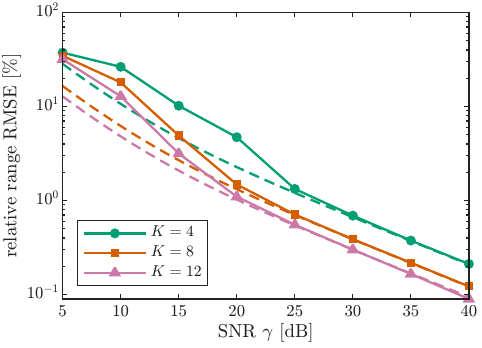}
\caption{Range-acquisition accuracy (Algorithm~\ref{alg:acq}),
$\NF = 10$, $\rho = 1/3$. Relative RMSE of $\widehat D$ versus SNR for
$K = 4, 8, 12$ probes; solid with markers: adaptive maximum-likelihood
(Monte-Carlo on the exact channel), dashed: Cram\'er--Rao bound
$\sqrt{1/J}/D$ \eqref{eq:fim}.}
\label{fig:acqrmse}
\end{figure}

\begin{figure}[!t]
\centering
\includegraphics[width=0.8\columnwidth]{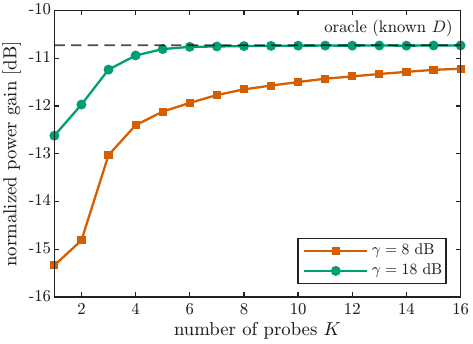}
\caption{Convergence of the acquired con-focusing gain to the oracle
(known $D$, dashed) versus the number of probes $K$, for two array SNRs
$\gamma = pNM/\sigma^{2}$ ($\NF = 10$, $\rho = 1/3$); the oracle is
reached in about five probes at high SNR.}
\label{fig:acqgain}
\vspace{-.5cm}
\end{figure}

Simulations use the exact spherical channel \eqref{eq:channel} (no
Fresnel approximation) on half-wavelength ULAs at a $30$~GHz carrier
($\lambda = 10$~mm), with up to $256$ elements per array and apertures
$L_{\mathrm{t}}, L_{\mathrm{r}}$ set per figure by $(\NF, \rho)$;
analysis curves are the closed forms of
Sections~\ref{sec:gains}--\ref{sec:confocus}. Lines denote analysis,
markers simulation, throughout.

Fig.~\ref{fig:gains} validates Theorems~\ref{thm:closedform}
and~\ref{thm:confocus}. The exact-channel markers fall on the
closed-form curves over the whole range of $\NF$; con-focusing rides
the rank-one eigenbound, the exact $\sigma_{\max}^2$, while focusing
decays as $\NF^{-2}$. Steering is the lone aperture-ratio-dependent
scheme: shown for $\rho = 0.9, 0.5, 0.1$, its gap to the bound widens
as the apertures grow unequal, and it overtakes focusing only past a
moderate $\NF$; at $\rho = 1$ it is identical to con-focusing.

Fig.~\ref{fig:boundary} maps the steering-to-focusing power-gain ratio
$10\log_{10}(\gff/\gfoc)$ over the $(\NF,\rho)$ plane, so the comparison
is read in decibels rather than as a binary verdict. The zero contour is
the crossover separatrix $\NFs(\rho)$ of Corollary~\ref{cor:map}:
focusing wins to its left and steering to its right, by the margin shown,
which exceeds $10$~dB deep in the near field. The separatrix recedes to
ever larger $\NF$ as $\rho \to 0$ (the MISO edge, where focusing always
wins) and tightens to its minimum at $\rho = 1$ (equal apertures); the
ripples in the contours near $\rho \approx 0.2$ and $\rho \approx 0.11$
are the narrow multi-crossing windows of Corollary~\ref{cor:map}.

Fig.~\ref{fig:landscape} shows the gain, normalized to the eigenbound
as $\NF G$ (optimum at $0$~dB), over the curvature plane at
$\NF = 10$ for shrinking $\rho$: only the confocal ridge of
Theorem~\ref{thm:landscape} survives, focusing on the user sits in a
trough, and the con-focal centre is the peak; as $\rho \to 0$ the
landscape stretches and degenerates toward the MISO limit, where only
the transmit curvature $q_{\mathrm{t}}$ matters.

Fig.~\ref{fig:eig} is the constructive half of the paper. Plotted as
$\NF G$, the con-focusing pair rides the rank-one eigenbound toward its
leading constant ($\NF G \to 1$), steering saturates a factor $\rho$
below, and focusing collapses ($\NF G \to 0$), confirming the three
regimes of Theorem~\ref{thm:confocus}.

Figs.~\ref{fig:acqrmse} and~\ref{fig:acqgain} validate the blind
acquisition of Algorithm~\ref{alg:acq}. The relative range error of the
adaptive maximum-likelihood estimator (Fig.~\ref{fig:acqrmse}, markers)
tracks the Cram\'er--Rao bound $\sqrt{1/J}/D$ \eqref{eq:fim} as the SNR
grows, for $K \in \{4, 8, 12\}$ probes: the estimator is efficient and
the accuracy is set by the Fisher information, with each probe adding a
term. The con-focusing gain built on the estimate $\widehat D$
(Fig.~\ref{fig:acqgain}) then climbs to the oracle (known $D$) within a
handful of probes, about
five at high SNR, so the protocol reaches the optimum from received
power alone, with no ranging and no geometry exchange. Each probe is
measured at its own SNR $\gamma G$, set by its array gain through
\eqref{eq:varG}; as the beam focuses this SNR rises, so the noisiest
measurement is the early far-field probe and the cleanest are the
con-focal ones near the peak, and the acquisition is self-accelerating.

\section{Conclusion}\label{sec:conclusion}

Near-field beamfocusing, the spherical-wavefront phase matched to the
user, is the default prescription for extremely large-scale
multiple-input multiple-output (XL-MIMO) arrays at millimeter-wave and
sub-terahertz carriers, yet its justification is invariably drawn for a
large array serving a single-antenna, multiple-input single-output
(MISO) user. We generalized near-field beamforming to the single-user
line-of-sight (LoS) MIMO link in which the user, too, carries an
extended aperture, by characterizing the entire family of focused beam
pairs between the two apertures, of which focusing on the user and
far-field steering are the extreme members. The comparison collapses
onto two dimensionless quantities, the link Fresnel number $\NF$ and
the aperture ratio $\rho$: focusing wins only below a crossover
$\NFs(\rho)$, the universal value $1.947$ for equal apertures and
receding to infinity as $\rho \to 0$, while steering wins beyond it,
the signal-to-noise ratio (SNR) gap widening by ten decibels per decade
of $\NF$. The celebrated focusing gain is thus the
vanishing-receive-aperture corner of an asymmetry map, not a near-field
universal.

The order-optimal unit-modulus member of the family at every $\NF$ is a
geometry-only \emph{con-focusing} law: instead of focusing each
aperture on the other, both are aimed at the common axial point from
which they subtend equal angles. It requires no channel state
information (CSI), attains the rank-one eigenbound in leading constant
as $\NF$ grows, so that amplitude tapering and full CSI add only a
vanishing margin, degenerates to plain steering for equal apertures,
tolerates ranging errors up to the classical depth of focus, and is
acquirable blindly, from received power alone, within a single
beam-refinement round. In short, in a true near-field MIMO link one
should fill the receive aperture, not focus on the user.

The analysis is confined to single-stream transmission over uniform
linear arrays (ULAs) in the paraxial regime. Natural extensions are
multi-stream operation, where the same confocal geometry should
organize the near-field multiplexing modes; planar and non-uniform
apertures; and the multiuser and wideband regimes of XL-MIMO.

\section*{Acknowledgment}
The authors would like to thank Prof.~Umberto Spagnolini for his
insightful advice and valuable suggestions.

\bibliographystyle{IEEEtran}
\bibliography{bibliography}

\end{document}


\title{Supplementary Material for\\``From Focusing to Con-Focusing:
Optimal Power Transfer in Line-of-Sight Near-Field MIMO''}
\author{Marouan~Mizmizi,~\IEEEmembership{Member,~IEEE,}
and~Sanzhar~Yergaliev,~\IEEEmembership{Student~Member,~IEEE}}
\maketitle

This document collects the proofs of the results of the main paper.
Cross-references to theorems, equations, and sections without the
``S'' prefix refer to the main paper; objects internal to this
supplement are prefixed with ``S''.

\appendices

\section{Auxiliary Oscillatory-Integral Estimates}\label{app:aux}

Every gain in the paper is the squared modulus of an oscillatory
integral over the normalized aperture $I = [-\tfrac12,\tfrac12]$. We
collect, once, the three estimates used throughout, in the Landau
notation of Section~\ref{sec:model}.

\begin{lemma}\label{lem:osc}
Let $b,\mu \in \mathbb{R}\setminus\{0\}$.
\begin{enumerate}
\item[(i)] For every $\beta \in \mathbb{R}$,
$\bigl|\int_{I} e^{\,j(bs^{2}+\beta s)}\,ds\bigr| \le 4/\sqrt{|b|}$.
\item[(ii)] $\int_{\mathbb{R}} e^{\,j\mu s^{2}}ds =
\sqrt{\pi/|\mu|}\,e^{\,j(\pi/4)\operatorname{sgn}\mu}$, and truncating
to $[a,b]$ with $0 \notin [a,b]$ changes the value by
$\bigO\bigl(1/(\mu a)\bigr)$.
\item[(iii)] For $\Phi(\bm q) = \bm q^{\top}\Omega\bm q$ with
$\Delta = \det\Omega \neq 0$, eigenvalues $\mu_{1},\mu_{2}$, and
signature $\sigma$,
\begin{equation}\label{eq:planeval}
\int_{\mathbb{R}^{2}} e^{\,j\Phi(\bm q)}\,d\bm q
= \frac{\pi}{\sqrt{|\Delta|}}\,e^{\,j(\pi/4)\sigma},
\end{equation}
and restricting the domain to $I^{2}$ changes the modulus by
$\bigO\bigl(\min_{i}|\mu_{i}|^{-1/2}\bigr)$.
\end{enumerate}
\end{lemma}

\begin{IEEEproof}
Bound (i) is the Van der Corput second-derivative test
\cite[Ch.~VIII]{Stein1993}. Bound (ii) is the Fresnel integral together
with its stationary-phase truncation \cite[Ch.~3]{Olver1974}; the
truncation error is one integration by parts,
\begin{equation}\label{eq:ibp}
\int_{a}^{b} e^{\,j\mu s^{2}}\,ds
= \left[\frac{e^{\,j\mu s^{2}}}{2j\mu s}\right]_{a}^{b}
+ \frac{1}{2j\mu}\int_{a}^{b}\frac{e^{\,j\mu s^{2}}}{s^{2}}\,ds
= \bigO\!\left(\frac{1}{\mu a}\right).
\end{equation}
Bound (iii) follows by diagonalizing $\Omega = R\,\mathrm{diag}(
\mu_{1},\mu_{2})\,R^{\top}$ with $R$ a rotation: the substitution
$\bm s = R^{\top}\bm q$ has unit Jacobian and separates the integral
into two copies of (ii), with $|\mu_{1}\mu_{2}| = |\Delta|$, while the
difference between the square $I^{2}$ and the plane is the edge slices,
each bounded through (ii) by $\bigO(\min_{i}|\mu_{i}|^{-1/2})$.
\end{IEEEproof}

\section{Closed-Form Gains and the Crossover}\label{app:proofs-gains}

\begin{IEEEproof}[Proof of Lemma~\ref{lem:delta}]
The field integral is $I = \int_{I^{2}} e^{\,j\Phi}$ with
$\Phi(\bm q) = \bm q^{\top}\Omega\bm q$ and $\det\Omega = \Delta$
(\eqref{eq:quadform}--\eqref{eq:Omega}). For $\Delta \neq 0$ the
gradient $\nabla\Phi = 2\Omega\bm q$ vanishes only at the aperture
centre $\bm q = 0$, where $\Phi = 0$: a single stationary point, so the
gain is the coherent contribution of that centre. Lemma~\ref{lem:osc}(iii)
evaluates it: $|I| = \pi/\sqrt{|\Delta|}$ up to the edge correction
$\bigO(\min_{i}|\mu_{i}|^{-1/2})$, whence
\begin{equation}
G = |I|^{2} = \Theta\bigl(|\Delta|^{-1}\bigr),
\end{equation}
with constant $\pi^{2}$ as $\min_{i}|\mu_{i}| \to \infty$. When
$\Delta \to 0$ one eigenvalue vanishes, the coherent band outgrows the
aperture, and \eqref{eq:planeval} no longer applies; $G$ is then capped
by the largest squared singular value of the kernel, the rank-one
eigenbound of order $\NF^{-1}$, attained on the locus
(Theorem~\ref{thm:confocus}).
\end{IEEEproof}

\begin{IEEEproof}[Proof of Theorem~\ref{thm:regions}]
Throughout, $\Si$ and the Fresnel function $F = C - jS$ are taken in
their standard small- and large-argument expansions.

\emph{(a) Far field ($\NF \to 0$).} The small-argument expansions
$\Si(z) = z - z^{3}/18 + \bigO(z^{5})$ and
$\varphi(\nu) = 1 + \bigO(\nu)$ (with $\nu_{+} - \nu_{-} = \NF$) give
\begin{equation}
\gfoc = 1 - \tfrac{\pi^{2}}{36}\NF^{2} + \bigO(\NF^{4}),
\qquad
\gff = 1 + \bigO(\nu_{+}),
\end{equation}
both tending to $1$.

\emph{(b) MISO edge ($\rho \ll 1$).} The first claim is
\eqref{eq:gfoc-cf} at $\NF \lesssim 1$. For the second, the difference
quotient in \eqref{eq:gff-cf} converges to the derivative
\eqref{eq:dnuphi} at $\nu_{0}$, whose modulus lies in
$\tfrac{1}{2\sqrt{\nu_0}}[1 - \tfrac{2}{\pi}\nu_0^{-1/2},\,
1 + \tfrac{2}{\pi}\nu_0^{-1/2}]$; squaring gives the
$\Theta(\nu_{0}^{-1})$ law with leading constant $\tfrac14$.

\emph{(c) Deep near field ($\NF \to \infty$).} The large-argument form
$\nu\varphi(\nu) = \sqrt{2\nu}\,F(\sqrt{2\nu}) - (1-e^{-j\pi\nu})/(j\pi)$,
the Fresnel bound
$\nu_{\pm}\varphi(\nu_{\pm}) = \tfrac{1-j}{\sqrt{2}}\sqrt{\nu_{\pm}} + \bigO(1)$,
and the exact identity
$\sqrt{\nu_{+}} - \sqrt{\nu_{-}} = L_{\mathrm{r}}/\sqrt{\lambda D} = \sqrt{\rho\NF}$
give
\begin{equation}
\gff^{1/2} = \frac{\sqrt{\nu_{+}} - \sqrt{\nu_{-}}}{\NF} + \bigO(\NF^{-1})
= \sqrt{\frac{\rho}{\NF}} + \bigO(\NF^{-1}),
\end{equation}
so $\gff = \rho/\NF + \lilo(\NF^{-1})$. The cancellation is essential:
for $\rho < 1$ the $\nu_{-}$ term removes half of the $\sqrt{\nu_{+}}$
contribution, and dropping it would give
$\bigl(\tfrac{1+\rho}{2\sqrt{\rho}}\bigr)^{2}\NF^{-1}$, which exceeds
the eigenbound of Theorem~\ref{thm:confocus}(ii) and is therefore
impossible. For focusing,
$\Si(z) = \tfrac{\pi}{2} - \cos z/z + \bigO(z^{-2})$ gives
\begin{equation}
\gfoc = \NF^{-2}\bigl(1 + \bigO(\NF^{-1})\bigr).
\end{equation}

\emph{(d) The crossover at $\rho = 1$.} The crossover $\NFs$ has no
closed form, and $\delta = \gfoc^{1/2} - \gff^{1/2}$, which shares the
sign of $\gfoc - \gff$, is built from the oscillatory $\Si$ and $F$ and
is not monotone, so uniqueness needs more than a sign of the
derivative. We bound $\delta$ analytically away from the crossover and
certify the single sign change in between; write
$\sigma = (1+\rho^{2})/(2\rho) \ge 1$. For small $\NF$, the alternating
series of $\Si$ and the entire series of $\nu\varphi$, with explicit
tail bounds, give
\begin{equation}\label{eq:smallNF}
\delta = \frac{(2\sigma^{2}-1)\,\pi^{2}}{180}\,\NF^{2} + R,
\qquad |R| \le 0.35\,(\sigma\NF)^{3},
\end{equation}
the cubic majorant constant $0.35$ being the certified sum of the
leading series tails. The cubic stays below the quadratic, hence
$\delta > 0$, while $(\sigma\NF) < (2\sigma^{2}-1)\pi^{2}/(63\sigma)$,
i.e. for $\sigma\NF \le 0.157$ (worst case $\sigma = 1$). For large
$\NF$, the bounds $|F - \tfrac{1-j}{2}| \le 2/(\pi z)$ and
$\Si \le 1.852$ give $\gff^{1/2} \ge \sqrt{\rho/\NF} - 8/(\pi\NF)$ and
$\gfoc^{1/2} \le 1.179/\NF$, so $\delta < 0$ once
$\sqrt{\rho\NF} > 8/\pi + 1.179$, i.e. $\NF \ge 13.9/\rho$, tightening
to $\NF \ge 6.02$ at $\rho = 1$ (where $\nu_{-} = 0$ halves the Fresnel
tail to $4/\pi$). On the remaining interval $[0.157, 6.02]$ at
$\rho = 1$, a verified computer-assisted certificate, Arb ball
arithmetic fed by the analytic majorants above for $|\delta'|$ and
$|\delta''|$, isolates a single sign change confined to
$[1.94698, 1.94742]$, with $\delta' < 0$ throughout and every enclosure
rigorously accounting for rounding; the same certificate confines the
zeros for $\rho \in [0.1, 1]$ (Corollary~\ref{cor:map}). The
certificate is provided as supplementary material
(\texttt{certify\_crossover.py}).
\end{IEEEproof}

\section{Discretization and Path Loss}\label{app:discrete}

\begin{proposition}[Sampling condition]\label{prop:sampling}
With $d_{\mathrm{t}} = L_{\mathrm{t}}/(N-1)$, $d_{\mathrm{r}} =
L_{\mathrm{r}}/(M-1)$, the sums
\eqref{eq:gfoc-sum}--\eqref{eq:gff-sum} agree with their continuous
limits up to a relative error $\bigO(1/N) + \bigO(1/M)$ provided each
residual phase is sampled above Nyquist. The binding term is the
steering quadratic phase, requiring $N > N_{\mathrm{t}} + \NF$ and
$M > N_{\mathrm{r}} + \NF$; for equal apertures ($\rho = 1$,
$N_{\mathrm{t}} = N_{\mathrm{r}} = \NF$) this is $N, M > 2\NF$. If
violated, grating (alias) foci appear and the continuous comparison
fails.
\end{proposition}

\begin{IEEEproof}
The quadrature is exact once every residual phase is sampled above
Nyquist; the proof identifies the fastest phase and reads off the rate.
The instantaneous frequency of the steering quadratic phase
$e^{-j\pi(x-y)^{2}/(\lambda D)}$ in $x$ is $|x-y|/(\lambda D) \le
c_{+}/(\lambda D)$, $c_{+} = (L_{\mathrm{t}}+L_{\mathrm{r}})/2$, larger
than the $|y|/(\lambda D) \le L_{\mathrm{r}}/(2\lambda D)$ of the
bilinear kernel $e^{jkxy/D}$; it is therefore binding. Sampling above
twice it, $1/d_{\mathrm{t}} > 2c_{+}/(\lambda D)$, gives $N - 1 >
L_{\mathrm{t}}(L_{\mathrm{t}}+L_{\mathrm{r}})/(\lambda D) =
N_{\mathrm{t}} + \NF$ (the bilinear kernel alone would give only
$\NF$), and symmetrically $M - 1 > N_{\mathrm{r}} + \NF$. Poisson
summation \cite[Ch.~VII]{Stein1993} then places all alias images
outside the band, and the Euler--Maclaurin formula
\cite[Ch.~8]{Olver1974} bounds the quadrature error by $\bigO(1/N) +
\bigO(1/M)$. For half-wavelength spacing the condition holds
automatically under Assumption~\ref{ass:paraxial} (it amounts to
$L_{\mathrm{t}} + L_{\mathrm{r}} < 2D$).
\end{IEEEproof}

\begin{proposition}[Path loss, uniform in $\NF$]\label{prop:pathloss}
Replace the unit-modulus entries by $H_{mn} = (D/d_{nm})
e^{-jkd_{nm}}$ and set $\eta \le (L_{\mathrm{t}} +
L_{\mathrm{r}})^{2}/(8D^{2})$. Then the relative perturbations obey,
uniformly over the $\NF$ sweep, $|\Delta\gff|/\gff \le C_{1}\eta$,
$|\Delta\gfoc|/\gfoc \le C_{2}\eta(1 + \ln_{+}\NF)$, and
$|\Delta\geig|/\geig \le C_{3}\eta$, with absolute constants.
\end{proposition}

\begin{IEEEproof}
Write the amplitude taper as $A(s) = (1+s^{2}/D^{2})^{-1/2}$, so that
$\eta$ bounds its deviation from unity over the aperture; we bound each
gain by the tool matched to its structure.
\emph{Steering.} The trapezoid collapse applies to the perturbed
integrand, and Lemma~\ref{lem:osc}(i) localizes the perturbation to the
same coherent strip $|s| \lesssim \sqrt{\lambda D}$ that carries the
main term, giving $|\Delta\gff^{1/2}| \le C\eta\sqrt{\lambda
D}/L_{\mathrm{t}}$, of the same order as $\gff^{1/2}$ itself.
\emph{Focusing.} Split at the focal-spot scale $|x| \le \lambda
D/L_{\mathrm{r}}$ and integrate by parts in $y$ outside it; the
harmonic integral yields the $\ln_{+}\NF$ factor.
\emph{Eigenbound.} Extend $A - 1$ to a $C^{2}$
compactly supported function and expand it in modulations:
$\Delta\bH = \tfrac{1}{2\pi}\int \hat a(\omega)\,
\mathrm{D}_{x}(\omega)\bH\mathrm{D}_{y}(\omega)^{H} d\omega$ with
unitary diagonal $\mathrm{D}$'s, so $\|\Delta\bH\|_{2} \le
\tfrac{1}{2\pi}\|\hat a\|_{1}\sigma_{\max} \le C\eta\,\sigma_{\max}$
and Weyl's perturbation inequality \cite[Ch.~7]{HornJohnson2013}
concludes; no lower bound on $\geig$ is required. All three bounds control the field-level (square-root)
perturbation; the relative power perturbation is at most twice as
large to first order, which the constants absorb.
\end{IEEEproof}

\section{The Confocal Landscape and the Con-Focusing Law}
\label{app:landscape}

The field integral \eqref{eq:gAB} of Section~\ref{sec:gains}
generalizes, with transverse offsets $(\alpha, \beta)$, to
$G(A,B,\alpha,\beta) = |I(A,B,\alpha,\beta)|^{2}$,
\begin{equation}\label{eq:gABoff}
I(A,B,\alpha,\beta) = \iint_{[-\frac12,\frac12]^{2}}
e^{j\Phi(u,v)}\,du\,dv ,
\end{equation}
with $\Phi = Au^{2} + Bv^{2} + 2\pi\NF uv + \alpha u + \beta v$; as in
Appendix~\ref{app:proofs-gains}, all bounds are derived on $I$ and
squared at the end.
The quadratic form degenerates iff $AB = \pi^{2}\NF^{2}$, which is
equivalent to $r_{\mathrm{t}} + r_{\mathrm{r}} = D$ (signed radii):
the confocal ridge. Steering sits at $(A,B) =
(-\pi N_{\mathrm{t}}, -\pi N_{\mathrm{r}})$, which satisfies the
ridge condition; focusing sits at $(0,0)$, which does
not. Conjugating the integrand and flipping $v$ shows
$G(A,B,\alpha,\beta) = G(-A,-B,-\alpha,\beta)$, which proves
Theorem~\ref{thm:landscape}(d) (FF equals the midpoint confocal pair
$(\pi N_{\mathrm{t}}, \pi N_{\mathrm{r}})$) and the degeneracy of
the converging and conjugate con-focusing branches
(Theorem~\ref{thm:confocus}).

\subsection{Off the ridge}
For $A \neq 0$, completing the square in $u$ writes
$I = \int_{I}\Psi(v)\,e^{\,jB_{\mathrm{eff}}v^{2}}\,dv$, with effective
curvature $B_{\mathrm{eff}} = B - \pi^{2}\NF^{2}/A$ and
$\Psi(v) = \int_{I}e^{\,jA(u-u_{\star})^{2}}\,du$ an incomplete Fresnel
integral about the stationary point
$u_{\star}(v) = -(\pi\NF v + \alpha/2)/A$. Split
$\Psi = \sqrt{\pi/|A|}\,e^{\,j\vartheta} + R$: the constant part times
the $v$-integral, bounded by Lemma~\ref{lem:osc}(i), produces the first
term of \eqref{eq:lemB}, while the tail obeys
$|R(v)| \le 2/\bigl(\pi|A|\,\mathrm{dist}(u_{\star}(v),\partial I)\bigr)$
by Lemma~\ref{lem:osc}(ii), and $\int_{I}|R(v)|\,dv$ is a harmonic
integral whose worst case is the window centred at $u_{\star} = 0$,
yielding the logarithmic remainder. Uniformly in $(\alpha, \beta)$,
\begin{equation}\label{eq:lemB}
|I| \le \frac{4\sqrt{\pi}}{\sqrt{|AB - \pi^{2}\NF^{2}|}}
+ \frac{4}{\pi\NF}\Bigl(1 + \ln_{+}\frac{3\pi\NF}{2\sqrt{|A|}}\Bigr),
\end{equation}
whose square is Theorem~\ref{thm:landscape}(c). The case $A = 0$,
$B \neq 0$ (and $B = 0$, $A \neq 0$) follows by the symmetry
$G(A,B,\alpha,\beta) = G(B,A,\beta,\alpha)$, reducing the $v$- (resp.\
$u$-) integral first; the case $A = B = 0$ (focusing family, any
offset) is handled directly, giving
$|I| \le \tfrac{2}{\pi\NF}(1 + \ln\tfrac{\pi\NF}{2})$, hence
$G = \bigO(\NF^{-2}\ln^{2}\NF)$.

\subsection{On the ridge}
Parametrize $A = \pi\NF\tau$, $B = \pi\NF/\tau$. The quadratic form
is the exact square $\tfrac{\pi\NF}{\tau}(\tau u + v)^{2}$;
substituting $s = \tau u + v$ collapses \eqref{eq:gAB} to a
one-dimensional Fresnel integral against a trapezoidal window, and a
stationary-phase evaluation with explicit remainder gives
\begin{equation}\label{eq:ridge}
G_{\mathrm{ridge}}(\tau) = \frac{\min(1,\tau)^{2}}{\tau\NF}
\Bigl(1 + O\!\Bigl(\frac{\ln\NF}{\sqrt{\NF}}\Bigr)\Bigr):
\end{equation}
every ridge pair is order-optimal with power-gain constant
$\kappa(\tau) = \min(\tau, 1/\tau) \le 1$, maximized at
$\tau = 1$, which is the equal-subtense pair \eqref{eq:law}; steering corresponds to $\tau = 1/\rho$, with constant
$\rho$.

\begin{IEEEproof}[Proof of Theorem~\ref{thm:confocus}]
\emph{(i) The law attains the ridge crest.} This is \eqref{eq:ridge} at
$\tau = 1$, where $\kappa(\tau) = \min(\tau,1/\tau)$ is maximal.

\emph{(ii) The crest meets the eigenbound.} Any unit-norm pair obeys
the rank-one (SVD) bound
\begin{equation}
G = \frac{|\wrx^{H}\bH\wt|^{2}}{NM}
\le \frac{\sigma_{\max}(\bH)^{2}}{NM},
\end{equation}
with equality at the dominant singular vectors; it thus bounds every
rank-one scheme, amplitude-tapered or channel-aware. Factor
$\bH = \bm{D}_{\mathrm{r}}\bm{K}\bm{D}_{\mathrm{t}}$,
where $\bm{D}_{\mathrm{t}}, \bm{D}_{\mathrm{r}}$ are diagonal
unit-modulus matrices carrying the separable quadratic phases
$e^{-jkx_{n}^{2}/2D}$, $e^{-jky_{m}^{2}/2D}$ and $\bm{K}_{mn} =
e^{\,jkx_{n}y_{m}/D}$ is the bilinear kernel; unitary diagonal factors
leave singular values unchanged, so $\sigma_{\max}(\bH) =
\sigma_{\max}(\bm{K})$. The continuous kernel
$\bm{K}(u,v) = e^{\,j\pi\NF\,uv}$ on $[-\tfrac12,\tfrac12]$ is, up to
the interval scaling, exactly the finite-Fourier (band-limiting)
operator whose singular functions are the prolate spheroidal wave
functions \cite{SlepianPollak1961,BoydKogelnik1962,Borgiotti1966}; with
time--bandwidth $c = \pi\NF/2$ its largest singular value is the
prolate $\lambda_{0}(c)$, $1 - \lambda_{0}(c) = \bigO(\sqrt{c}\,e^{-2c})$.
For continuous apertures this gives
$\geig = \NF^{-1}\lambda_{0}(c)$. On a discrete array, by
Proposition~\ref{prop:sampling} the discrete $\bm{K}$ agrees with this
operator up to a relative $\bigO(1/N) + \bigO(1/M)$ once $N, M > 2\NF$, so
\begin{equation}
\geig = \NF^{-1}\bigl(1 + \bigO(\sqrt{\NF}\,e^{-\pi\NF})
+ \bigO(1/N) + \bigO(1/M)\bigr),
\end{equation}
the exponential being the prolate gap and the $\bigO(1/N)$ the
discretization.

\emph{(iii) Ratio.} Dividing (i) by (ii) and invoking
\eqref{eq:deepNF} gives the deep-field ratio.
\end{IEEEproof}

\begin{IEEEproof}[Proof of Lemma~\ref{lem:dof}]
The argument has four steps: linearize the curvature mismatch in
$\varepsilon$, convert it into the off-ridge discriminant, invoke the
off-ridge upper bound for the outer ($3$-dB-loss) threshold, and the
exact perturbed ridge reduction for the matching inner (persistence)
threshold; the absolute-error form follows by a change of variables.

\emph{Step 1 (curvature mismatch).} The focal phases of the con-focusing pair are computed for the estimated range $\widehat D = (1+\varepsilon)D$,
so the applied per-aperture curvatures scale as $1/\widehat D$ while
the true channel curvature scales as $1/D$. Writing
$1/\widehat D = (1-\varepsilon+\bigO(\varepsilon^{2}))/D$, the residual
quadratic-phase curvatures of \eqref{eq:gAB}, which vanish on the
confocal ridge at $\varepsilon=0$, are
\begin{align}
A &= \pi\NF - \pi N_{\mathrm{t}}(1+\rho)\,\varepsilon
+ \bigO(\varepsilon^{2}), \label{eq:detune}\\
B &= \pi\NF - \pi N_{\mathrm{r}}(1+\rho^{-1})\,\varepsilon
+ \bigO(\varepsilon^{2}).\nonumber
\end{align}

\emph{Step 2 (discriminant).} Substituting \eqref{eq:detune} into the
ridge discriminant and using the identity
$N_{\mathrm{t}}(1+\rho)+N_{\mathrm{r}}(1+\rho^{-1})
= (L_{\mathrm{t}}+L_{\mathrm{r}})^{2}/(\lambda D) = 4\nu_{+}$,
\begin{equation}\label{eq:disc}
\begin{aligned}
|AB - \pi^{2}\NF^{2}|
&= \pi\NF\bigl[\pi N_{\mathrm{t}}(1+\rho)
   + \pi N_{\mathrm{r}}(1+\rho^{-1})\bigr]|\varepsilon|\\
&\quad + \bigO(\varepsilon^{2})
= 4\pi^{2}\NF\,\nu_{+}|\varepsilon| + \bigO(\varepsilon^{2}):
\end{aligned}
\end{equation}
the pair leaves the ridge at a rate set precisely by the
joint-aperture Fresnel number $\nu_{+}$.

\emph{Step 3 (outer threshold; necessity).} By the off-ridge bound
\eqref{eq:lemB}, $g := G^{1/2} \le 4\sqrt{\pi}\,
|AB-\pi^{2}\NF^{2}|^{-1/2} + \softO(\NF^{-1})$. The optimal gain is
$G_{\star} = \NF^{-1}(1+\lilo(1))$ by \eqref{eq:ridge} at $\tau=1$, so
falling more than $3$~dB below it means $G \le \tfrac14\NF^{-1}$,
i.e.\ $g \le \tfrac12\NF^{-1/2}$. The leading term of \eqref{eq:lemB}
reaches this value when $|AB-\pi^{2}\NF^{2}| = 64\pi\NF$, which by
\eqref{eq:disc} is $|\varepsilon| = 16/(\pi\nu_{+})$. For
$|\varepsilon|$ beyond this the upper bound \emph{forces} the loss,
giving the outer threshold.

\emph{Step 4 (inner threshold; persistence).} That the gain genuinely
\emph{persists} near the optimum below an $\bigO(1/\nu_{+})$ scale
requires a matching lower bound, obtained from the exact ridge
reduction perturbed to second order. For general $\rho$ the detuned form $Q = \bm{q}^{H}\Omega\bm{q}$ in
\eqref{eq:gAB} has trace $A+B = 2\pi\NF - 4\pi\nu_{+}\varepsilon +
\bigO(\varepsilon^{2})$ and determinant \eqref{eq:disc}, hence
eigenvalues $\mu_{1} = 2\pi\NF + \bigO(\varepsilon)$ (the on-ridge
curvature) and $\mu_{2} = \Delta/\mu_{1} = -2\pi\nu_{+}\varepsilon +
\bigO(\varepsilon^{2})$ (the transverse curvature). To first order the
on-ridge direction is undetuned, while the transverse direction
acquires the quadratic phase $\mu_{2}w^{2}$ over an $\bigO(1)$ extent of
the normalized aperture; the gain therefore stays within $3$~dB of its
$\varepsilon=0$ value as long as $|\mu_{2}| = \bigO(1)$, i.e.\
$|\varepsilon| \le c_{1}/\nu_{+}$ for a constant $c_{1}>0$ bounded
uniformly in $\rho$, the inner threshold. For $\rho=1$ the
eigen-directions are $s=(u+v)/\sqrt2$, $w=(u-v)/\sqrt2$ and the integral
factors exactly,
\begin{equation}\label{eq:Qdiag}
Q = (2\pi\NF-\bar\delta)\,s^{2} - \bar\delta\,w^{2},
\qquad \bar\delta = 2\pi\nu_{+}\varepsilon,
\end{equation}
recovering the threshold with an explicit constant.

\emph{Absolute-error form.} With $\Delta D = \varepsilon D$ and
$\nu_{+} = (L_{\mathrm{t}}+L_{\mathrm{r}})^{2}/(4\lambda D)$, the
relative bound $|\varepsilon| \le c_{1}/\nu_{+}$ becomes
$|\Delta D| \le 4c_{1}\lambda\,(D/(L_{\mathrm{t}}+L_{\mathrm{r}}))^{2}$,
which is \eqref{eq:dof}; the two thresholds pin the implied constant
between $4c_{1}$ and $64/\pi$, and evaluating the closed-form gain at
the perturbed ridge places the sharp value $c_{0} \approx 3/\pi$ inside
this band.
\end{IEEEproof}

\section{Rotations: Tilt and Roll}\label{app:rotations}

The strategy is to reduce the off-boresight gains to the boresight ones
by projecting each aperture onto the broadside plane: an affine map
turns every boresight functional into the same functional of the
projected apertures, with $\NF \mapsto \NFe$ and the cross terms scaled
by $\cos\phi$; the two proofs then read off the tilt and the roll.

With the projector reduction of Section~\ref{sec:confocus} (the exact
identity $d^{2} = (D+a)^{2} + c^{2}$, longitudinal offset $a$,
transverse part $c^{2}$ of \eqref{eq:lawcos}), the paraxial expansion
carries a third-order remainder,
\begin{equation}\label{eq:offb-exp}
d = D + a + \frac{c^{2}}{2D} + \varepsilon_{3},
\qquad
|\varepsilon_{3}| \le \frac{c_{+}^{2}\tilde a}{2D(D - \tilde a)}
+ \frac{c_{+}^{4}}{8(D-\tilde a)^{3}},
\end{equation}
with $\tilde a = \tfrac12(L_{\mathrm{t}}|\sin\theta_{\mathrm{t}}| +
L_{\mathrm{r}}|\sin\theta_{\mathrm{r}}|)$,
$c_{+} = \tfrac12(L_{\mathrm{t}}\cos\theta_{\mathrm{t}} +
L_{\mathrm{r}}\cos\theta_{\mathrm{r}})$ the projected joint
half-aperture, and $\tilde\nu_{+} = c_{+}^{2}/(\lambda D)$ its Fresnel
number: off boresight, validity requires the additional cubic (coma)
condition $\pi\tilde\nu_{+}\tilde a/D \ll 1$.

Substituting \eqref{eq:lawcos} and conjugating the appropriate phases
leaves, in the projected coordinates, the two residual kernels
\eqref{eq:twistkernels}. The affine map $x \mapsto \tilde x =
x\cos\theta_{\mathrm{t}}$ (and $y \mapsto \tilde y$) pushes the
normalized uniform aperture measure onto the projected apertures
$L_{\mathrm{t}}\cos\theta_{\mathrm{t}}$, $L_{\mathrm{r}}
\cos\theta_{\mathrm{r}}$, so every boresight quantity is the same
functional of these two integrals with $\NF \mapsto \NFe$ and the
cross terms scaled by $\cos\phi$.

\begin{IEEEproof}[Proof of Corollary~\ref{cor:rotations}]
At $\phi = 0$, \eqref{eq:lawcos} reduces to $c^{2} = (\tilde x -
\tilde y)^{2}$ and \eqref{eq:twistkernels} become the focusing kernel
$e^{jk\tilde x\tilde y/D}$ and the steered quadratic phase
$e^{-jk(\tilde x - \tilde y)^{2}/(2D)}$, identical to the boresight integrands
\eqref{eq:gfoc-sum}--\eqref{eq:gff-sum} in $(\tilde x, \tilde y)$.
Every result of Theorems~\ref{thm:closedform}--\ref{thm:confocus} and
Corollary~\ref{cor:map} is a functional of these integrals and of the
prolate decomposition of the focusing kernel, whose time--bandwidth
product is $\pi\NFe/2$; hence they transfer verbatim under
\eqref{eq:angmap}. The crossover value follows from
Theorem~\ref{thm:regions} with $\NF \mapsto \NFe$, giving
$\NFs(\theta) = 1.947/(\cos\theta_{\mathrm{t}}\cos\theta_{\mathrm{r}})$.
\end{IEEEproof}

\begin{IEEEproof}[Proof of Proposition~\ref{prop:roll}]
For general $\phi$ the cross-coupling parameter of
\eqref{eq:twistkernels} is $\NFe\cos\phi$, so the confocal ridge of
the modified kernel is $AB = \pi^{2}(\NFe\cos\phi)^{2}$ in the
curvature coordinates \eqref{eq:AB} of the projected apertures, with
per-aperture numbers $N_{\mathrm{t}}^{\mathrm{eff}} =
(L_{\mathrm{t}}\cos\theta_{\mathrm{t}})^{2}/(\lambda D)$,
$N_{\mathrm{r}}^{\mathrm{eff}}$, and $\NFe =
\sqrt{N_{\mathrm{t}}^{\mathrm{eff}} N_{\mathrm{r}}^{\mathrm{eff}}}$.

\emph{(a)} Steering sets both focal depths to infinity,
i.e.\ $A = -\pi N_{\mathrm{t}}^{\mathrm{eff}}$ and $B = -\pi
N_{\mathrm{r}}^{\mathrm{eff}}$, whence $AB = \pi^{2}
N_{\mathrm{t}}^{\mathrm{eff}} N_{\mathrm{r}}^{\mathrm{eff}} =
\pi^{2}(\NFe)^{2}$ and the off-ridge discriminant is
\begin{equation}\label{eq:twistdisc}
\begin{aligned}
AB - \pi^{2}(\NFe\cos\phi)^{2}
&= \pi^{2}(\NFe)^{2}(1 - \cos^{2}\phi)\\
&= \pi^{2}(\NFe)^{2}\sin^{2}\phi .
\end{aligned}
\end{equation}
The off-ridge bound \eqref{eq:lemB} then gives
\begin{equation}
\begin{aligned}
G^{1/2}
&\le \frac{4\sqrt{\pi}}{\sqrt{\pi^{2}(\NFe)^{2}\sin^{2}\phi}}
   + \softO(\NF^{-1})\\
&= \frac{4}{\sqrt{\pi}\,\NFe|\sin\phi|} + \softO(\NF^{-1}),
\end{aligned}
\end{equation}
so for fixed $\phi \neq 0$, $G = \softO(\NF^{-2})$.

\emph{(b)} The focusing gain is \eqref{eq:gfoc-cf} evaluated at the
cross-coupling $\NFe\cos\phi$,
$\gfoc = [\tfrac{2}{\pi\NFe\cos\phi}\Si(\tfrac{\pi\NFe\cos\phi}{2})]^{2}$,
which is decreasing in its argument and hence increasing as $\phi$
grows; as $\phi \to \pi/2$ the argument $\to 0$, $\Si(z) = z + \bigO
(z^{3})$, and $\gfoc \to 1$: the kernel \eqref{eq:twistkernels} is
then separable and all $NM$ pairs combine coherently. Dividing by the
upper bound of (a), $\gfoc/\gff^{\mathrm{ub}} \asymp \tan^{2}\phi$; since
$\gff \le \gff^{\mathrm{ub}}$, focusing exceeding the bound implies
$\gfoc > \gff$ a fortiori, so the condition $\tan^{2}\phi = 1$ gives a
\emph{conservative} dominance angle $\phi_{c} \to 45^{\circ}$ as
$\NFe \to \infty$. Finite-$\NFe$ corrections, valid within the
paraxial-plus-coma regime of this appendix, place it near
$50^{\circ}$--$55^{\circ}$ at moderate Fresnel numbers.

\emph{(c)} The ridge $AB = \pi^{2}(\NFe\cos\phi)^{2}$ is nonempty; its
equal-subtense ($\tau = 1$) member has $A = B = \pi\NFe\cos\phi$.
With $A = \pi N_{\mathrm{t}}^{\mathrm{eff}}(D/r_{\mathrm{t}} - 1)$ from
\eqref{eq:AB} and $N_{\mathrm{t}}^{\mathrm{eff}} = \NFe$ at $\rho = 1$,
this reads $D/r_{\mathrm{t}} - 1 = \cos\phi$, i.e.\
$r_{\mathrm{t}}^{\star} = r_{\mathrm{r}}^{\star} = D/(1+\cos\phi)$; the
general asymmetric pair solves $A = B = \pi\NFe\cos\phi$ for
$(r_{\mathrm{t}}, r_{\mathrm{r}})$. On this ridge \eqref{eq:ridge}
with $\NF \mapsto \NFe\cos\phi$ yields $G = \Theta((\NFe\cos\phi)^{-1})$
for $\NFe\cos\phi \gg 1$, matching the eigenbound order; once
$\NFe\cos\phi \lesssim 1$ (roll near $\pi/2$) the kernel is essentially
separable and $G = \Theta(1)$, so the roll-corrected pair dominates
both (a) and (b) at every $\phi$.
\end{IEEEproof}

\bibliographystyle{IEEEtran}
\bibliography{bibliography}